\documentclass[journal]{IEEEtran}
\usepackage{amsmath,dsfont,bbm,epsfig,amssymb,amsfonts,amstext,verbatim,amsopn,cite,subfigure,multirow,multicol,lipsum,xfrac}
\usepackage{amsthm}
\usepackage{mathtools,amsthm}
\usepackage{perpage}
\usepackage{balance}
\usepackage{url}
\usepackage{amsfonts}
\usepackage{epsfig}
\usepackage[font={small}]{caption}
\usepackage{psfrag}	
\usepackage{etoolbox}
\usepackage{algorithmicx}
\usepackage[Algorithm,ruled]{algorithm}
\usepackage{algpseudocode}
\usepackage{pifont}
\usepackage[utf8]{inputenc}
\usepackage[T1]{fontenc}  
\usepackage[nolist]{acronym}
\MakePerPage{footnote}
\usepackage{paralist}
\usepackage{enumitem}
\usepackage{bbm}
\usepackage{tikz}
\usepackage{pgfplots}
\usetikzlibrary{shapes,arrows}
\newtheorem{theorem}{Theorem}

\newcommand{\jj}{\mathrm{j}}
\newcommand{\Ex}[1]{ \mathcal{E} \left[  #1 \right]}
\newcommand{\dbc}[1]{ \left[  #1 \right]}
\newcommand{\tr}[1]{\mathrm{tr} \left\lbrace #1 \right\rbrace}
\newcommand{\set}[1]{ \left\lbrace #1 \right\rbrace}
\newcommand{\brc}[1]{ \left(  #1 \right)}
\newcommand{\abs}[1]{  \left\vert  #1 \right\vert}
\newcommand{\norm}[1]{  \left\Vert  #1 \right\Vert}

\newcommand{\btau}{\boldsymbol{\psi}}
\newcommand{\bbeta}{\boldsymbol{\beta}}
\newcommand{\her}{\mathsf{H}}
\newcommand{\xx}{{\mathrm{x}}}
\newcommand{\yy}{{\mathrm{y}}}
\newcommand{\trp}{{\mathsf{T}}}

\newcommand{\maE}{{\mathcal{E}}}

\newcommand{\mPhi}{\boldsymbol{\Phi}}
\newcommand{\mSigma}{\boldsymbol{\Sigma}}
\newcommand{\bmu}{\boldsymbol{\mu}}

\newcommand{\setC}{{\mathbb{C}}}
\newcommand{\mR}{{\mathbf{R}}}
\newcommand{\mC}{{\mathbf{C}}}
\newcommand{\mV}{{\mathbf{V}}}
\newcommand{\mI}{{\mathbf{I}}}
\newcommand{\setR}{{\mathbb{R}}}
\newcommand{\bh}{{\mathbf{h}}}

\newcommand{\bv}{{\mathbf{v}}}

\newtheorem{lemma}{Lemma}
\newtheorem{remark}{Remark}
\newtheorem{proposition}{Proposition}
\newtheorem{pproposition}{Pseudo-Proposition}

\begin{document}
\title{Channel Hardening of IRS-Aided Multi-Antenna Systems: How Should IRSs Scale?}

\author{
\IEEEauthorblockN{
Ali Bereyhi, \textit{Member, IEEE}, 
Saba Asaad, \textit{Member, IEEE}, 
Chongjun Ouyang, \textit{Student Member, IEEE}, 
Ralf R. M\"uller, \textit{Senior Member, IEEE}, 
Rafael F. Schaefer, \textit{Senior Member, IEEE} and 
H. Vincent Poor, \textit{Life Fellow, IEEE}}
\IEEEauthorblockA{
\thanks{ Saba Asaad, Ali Bereyhi and Ralf Müller are with the Friedrich-Alexander Universit\"at; emails: \{saba.asaad, ali.bereyhi,ralf.r.mueller\}@fau.de. Chongjun Ouyang is with the Beijing University of Posts and Telecommunications; email: dragonaim@bupt.edu.cn. H. Vincent Poor is with the Princeton University; email: poor@princeton.edu.}
\thanks{This work was supported by the German Research Foundation, Deutsche Forschungsgemeinschaft (DFG), under Grant No. MU 3735/7-1, and in part by the German Federal Ministry of Education and Research (BMBF) under the Grant 16KIS1242.}}}


\IEEEoverridecommandlockouts

\maketitle

\begin{abstract}
	$\;$Unlike active array antennas, intelligent reflecting surfaces (IRSs) are efficiently implemented at large dimensions. This allows for traceable realizations of large-scale \textit{IRS-aided} MIMO systems in which not necessarily the array antennas,~but the passive IRSs are large. It is widely believed that large IRS-aided MIMO settings maintain the fundamental features of massive MIMO systems, and hence they are the implementationally feasible technology for establishing the performance of large-scale MIMO settings. This work gives a rigorous proof to this belief. We show that using a large passive IRS, the end-to-end MIMO channel between the transmitter and the receiver always hardens, even if the IRS elements are strongly correlated. 
	
For the fading direct and reflection links between the transmitter and the receiver, our derivations demonstrate that as the number of IRS elements grows large, the capacity of end-to-end channel converges in distribution to a real-valued Gaussian~random~variable whose variance goes to zero. The order of this drop depends on how the physical dimensions of the IRS grow. We derive this order explicitly. Numerical experiments depict that the analytical asymptotic distribution almost perfectly matches the histogram of the capacity, even in practical scenarios.
	
	As a sample application of the results, we use the asymptotic characterization to study the dimensional trade-off between the transmitter and the IRS. The result is intuitive: For a given target performance, the larger the IRS is, the less transmit antennas are required to achieve the target. For an arbitrary ergodic and outage performance, we characterize this trade-off analytically. Our investigations demonstrate that using a practical IRS size, the target performance can be achieved with significantly small end-to-end MIMO dimensions.

\end{abstract}
\IEEEpeerreviewmaketitle
\section{Introduction}
Channel hardening is a property of massive \ac{mimo} systems indicating that the random fading process in a \ac{mimo} channel becomes a deterministic effective path-loss, as the  dimensions of the \ac{mimo} channel grow large at least at one side; see for instance \cite{hochwald2004multiple,bai2009rate,asaad2018massive}. From the information-theoretic point of view, this is the key property of massive \ac{mimo} systems which leads to their~significant~performance gains in various aspects. These gains are rather well-studied in the literature; e.g., \cite{bjornson2017massive, asaad2018optimal,narasimhan2014channel}. 

Despite the promising performance, massive \ac{mimo} systems are still considered to be implementationally intractable. This follows from the fact that deploying large active antenna arrays imposes high hardware cost and energy consumption, as well as unrealistic design requirements \cite{bjornson2016massive, bjornson2015massive, bereyhi2018stepwise}. The issue becomes more crucial in practice, since several recent studies~show~that the fundamental features of massive \ac{mimo} systems\footnote{These features are mainly described by the channel hardening.} do not hold with practical dimensions in several settings, e.g.,  \ac{mimo} systems with highly-correlated channels, scenarios with small scattering angular spread and cell-free networks; some detailed studies can be found in \cite{caire2018ergodic} and \cite{ngo2015cell}. 

\Ac{irs}-aided \ac{mimo} communication has been recently proposed as an implementationally-tractable solution to realize large-dimensional \ac{mimo} systems \cite{di2019smart, yu2020robust, bereyhi2020secure}. An \ac{irs} is a two-dimensional surface, with a large number of reconfigurable passive reflective elements. These elements are capable of tuning the wireless propagation dynamically to achieve various design objectives: For instance, in indoor millimeter-wave applications, by installing \acp{irs} out of the transmission site, e.g., on interior walls or room ceilings, the blockage issue of the millimeter-wave communication can be reduced significantly, due to the link established between the transmitter and the users through the \acp{irs} \cite{cao2020intelligent}. As other examples of design targets which are achieved by employing \acp{irs} in the system, we can name spatial interference suppression, accurate three-dimensional beamforming, and providing a more favorable propagation environment \cite{di2020smart, nadeem2019intelligent,zhang2021reconfigurable, wu2019towards}.

As mentioned, unlike classical large-scale \ac{mimo} technologies which require large active arrays of antennas, large \ac{irs}-aided \ac{mimo} systems, i.e., \ac{mimo} settings with large reflecting surfaces, are implementationally feasible. This follows the fact that \acp{irs} \textit{passively} reflect the incoming signals and do not employ radio frequency chains. As a result, the hardware cost and energy consumption at large scale decrease notably \cite{gong2020toward, wu2021intelligent}. Inspired by this appealing advantages, \acp{irs} are employed to a develop practically-tractable alternative designs for several large-scale \ac{mimo} technologies, such as massive \ac{mimo} \cite{li2019joint, wang2021massive}, cognitive radio \cite{xu2020resource, wu2021irs}, \ac{noma} \cite{ding2020simple, ding2020impact} and simultaneous wireless information and power transform systems \cite{tang2020joint, zargari2020energy}. The key motivation behind these designs comes from this widely-accepted belief: \textit{The key features of a large-scale \ac{mimo} technology can be achieved via a small-dimension \ac{mimo} setting assisted by large \acp{irs}.} Although under idealistic assumptions this belief seems to be intuitive, it is not easy to conclude its validity in practical~\ac{irs}-aided settings with highly-correlated \ac{irs} elements and multiple propagation links between the transmitter and the receiver. This work aims to give a rigorous proof for this believe by considering the fundamental channel hardening property in a practical \ac{irs}-aided multi-antenna setting.

\subsection{Main Objective and Contributions}
This work gives an answer to this intriguing question:~\textit{How does the end-to-end channel in \ac{irs}-aided multi-antenna systems harden, when only the \ac{irs} size grows large?} We answer this question by deriving the asymptotic distribution of the channel capacity expression in an \ac{irs}-aided fading \ac{miso} system\footnote{It is worth mentioning that the results readily extend to cases with multiple receive antennas. A \ac{miso} setting is assumed mainly to keep the derivations tractable. This assumption however impacts neither the analytical approach nor the final conclusions of this study.} whose number of transmit antennas is fixed and rather small, when the number of \ac{irs} elements grow asymptotically \textit{large}. Our derivations demonstrate that for any \ac{irs} covariance matrix, the channel capacity converges in distribution to a  real-valued Gaussian distribution whose mean increases unboundedly large, and whose variance tends to zero, as the \ac{irs} size grows asymptotically large.

The analytical derivations of this study gives the following answer to our target question: \textit{Even with a strongly-correlated \ac{irs}, the end-to-end channel hardens regardless of the transmit array size\footnote{In fact, it hardens asymptotically with only a single transmit antenna.}}. Interestingly, this property is achieved \textit{without} any need for frequent tuning of \ac{irs} elements.

In addition to the main results, this study has several other contributions which are briefly highlighted below:
\begin{itemize}
\item We characterize the distribution of the channel capacity for a generic \ac{irs}-aided fading \ac{miso} setting in which both direct and reflecting links are available between the transmitter and receiver. We further take into account the impact of line-of-sight channels and consider an arbitrary scaling of the \ac{irs} area and a general correlation among the reflecting elements. To the best of our knowledge,~the characterization of channel hardening for such setting has been left open in the literature; see for instance the recent studies in \cite{jung2020performance} and \cite{zhang2021outage}.
	\item For a given \ac{irs} and transmit array architecture,~we~determine the \textit{exact} distribution of the end-to-end \ac{snr} for an arbitrary choice of phase-shifts at the \ac{irs}. We give a universal upper-bound for the mean and variance of this distribution, and propose an effective choice of \ac{irs} phase-shifts that guarantees the end-to-end channel hardening.
	\item We validate our analytical derivations through numerical experiments. Our investigations show that the asymptotic distribution fits almost perfectly the histogram of channel capacity for practical \ac{irs} sizes\footnote{And even much smaller sizes of \ac{irs}.}. We further~conduct~several numerical experiments to investigate channel hardening for classical and extreme \ac{irs} correlation scenarios.
	\item To characterize the speed of channel hardening, we derive a lower-bound on the mean and a uniform upper-bound~on the variance of the channel capacity. Invoking the bounds, we show that for any covariance matrix~at~the~\ac{irs}~whose largest eigenvalue grows \textit{sub-linearly}\footnote{The definition of sub-linear growth becomes clear in the forthcoming sections.} with the \ac{irs} size, the limit of the variance, when the \ac{irs} size grows large, is zero. Using extreme-case investigations, we demonstrate that this conclusion is further extended to most extreme scenarios with the largest eigenvalue growing \textit{linearly} and the physical dimension of \ac{irs} being fixed.
	\item As a sample application, we use the main results to investigate the dimensional trade-off between the transmitter and the \ac{irs}. The result is intuitive: \textit{By increasing the \ac{irs} size, the transmit array size required to achieve a given target performance reduces}. We derive the trade-off curve analytically two performance metrics, namely the ergodic capacity and outage probability.
\end{itemize}

\subsection{Related Work}
Channel hardening in \ac{irs}-aided \ac{mimo} systems is studied in the literature for some restricted settings: In \cite{bjornson2020rayleigh}, channel hardening is discussed in an isotropic scattering environment considering an \ac{irs}-aided setting with a single-antenna transmitter and receiver, i.e., a \ac{siso} setting. The analysis is extended to scenarios with multiple \acp{irs} and \textit{uncorrelated Rician} fading in \cite{zhang2019analysis}. The concept of channel hardening is further addressed in an alternative way in \cite{ibrahim2021exact}, where the authors show that in an \ac{irs}-aided \ac{siso} system with generic Nakagami-$m$ fading channels, the end-to-end channel between the transmitter and the receiver, in the presence of a direct link, becomes nearly deterministic. The results of this study further demonstrate that increasing the number of \ac{irs} elements, as well as reducing the fading severity, enhances the channel hardening property. 

In addition to classical point-to-point settings, channel hardening has been further discussed in alternative scenarios. For instance, channel hardening is investigated in \cite{wang2021performance} for an \ac{irs}-aided \ac{siso} \ac{noma} system, where the \ac{irs} is assumed to be realized by the intelligent omni-surfaces technology, and the reflecting elements are considered to be correlated in general. The results show that for such settings, the average achievable rate converges asymptotically to that of uncorrelated channels, when the number of \ac{irs} elements goes to infinity. Another example is a \ac{mimo} setting with a fully stand-alone \ac{irs}-based transmitter. For such settings, the authors of \cite{jung2020performance} show that by growing the size of the transmitter unit, the user channels become deterministic and mutually orthogonal. 

Despite connections to the above literature, the most related lines of work to our study are those given in \cite{wang2020intelligent,wang2021massive} and \cite{zhang2021outage}: Starting with the study in \cite{wang2020intelligent}, the authors show that the conventional form of channel hardening is not valid for an \ac{irs}-aided \ac{mimo} setting. The result depicts that by increasing the number of transmit antennas, while keeping the number of \ac{irs} elements fixed, the hardening property does not hold. This finding follows from the fact that the channel between the \ac{irs} and transmitter is identical for all receivers. As a result, by growth of transmit array size, the randomness of the fading process grows large, and hence the effective end-to-end channel does not converge to a deterministic channel.  In their following work, i.e.,  \cite{wang2021massive}, the authors invoke the Lindeberg-Feller central theorem \cite{van2000asymptotic} to demonstrate that the channel hardening in \ac{irs}-aided systems occurs when the \textit{number of \ac{irs} elements} grows large\footnote{The difference of \cite{wang2021massive} to this work is illustrated shortly after.}. 

The study in \cite{zhang2021outage} derives the distribution of the input-output mutual information for an \ac{irs}-aided \ac{mimo} setting in which the communication is carried out \textit{only} through the \ac{irs}, i.e., there is no direct link between the transmitter and receiver. The analysis invokes random matrix theory to characterize the setting in an asymptotic regime, in which both the numbers of \ac{irs} elements and transmit antennas grow large with a fixed-ratio. The authors further utilize the results to derive a closed-form expression for the outage probability and the optimal values of \ac{irs} phase-shifts. 

Although the studies in \cite{zhang2021outage,wang2021massive} extend earlier results on channel hardening to a wider scope of settings, they consider several simplifying assumptions which impact the validity of the final conclusions in practical scenarios. For instance, they consider the channel coefficients of reflecting elements to be \textit{uncorrelated} and ignore the line-of-sight between the \ac{irs} and other terminals in the network. These assumptions are rather unrealistic, as in many use-cases of \ac{irs}-aided systems, the reflecting elements are assumed to be closely packed on the \ac{irs}, and the end-nodes are located in the line-of-sight of the \ac{irs} \cite{han2019large}. Moreover, \cite{zhang2021outage} assumes the complete blockage of the direct path between the transmitter and the receiver which despite its validity in some use-cases, restricts the applicability of the results and makes the comparison between a large \ac{irs}-aided \ac{mimo} system and a massive \ac{mimo} system unfair. The asymptotic regime considered in \cite{zhang2021outage}  is further impractical, as we are often interested in \ac{irs}-aided \ac{mimo} settings with large \ac{irs} dimensions, but rather small arrays at the transmitter and receiver sides. In this work, we deviate from these simplifying assumptions and characterize the asymptotic channel hardening principle for a generic \ac{irs}-aided scenario.

%


\subsection{Notation and Organization}
Scalars, vectors and matrices are shown by non-bold, bold lower-case, and bold upper-case letters, respectively. The notation $\mathbf{H}^{\mathsf{H}}$ indicates the transposed conjugate of $\mathbf{H}$. An $N\times N$ identity matrix is denoted by $\mathbf{I}_N$ and $\boldsymbol{1}_N$ is an $N\times N$ matrix of all-ones. We use the notation $\dbc{\mathbf{H}}_{nm}$ to refer to the entry~of $\mathbf{H}$ at the $n$-th row and $m$-th column. The function $\mathrm{Q}\brc{x}$ is the standard $\mathrm{Q}$-function, i.e.,
\begin{align}
	\mathrm{Q} \brc{x} = \frac{1}{\sqrt{2\pi}}\int_{x}^{+\infty} \mathrm{e}^{-\tfrac{u^2}{2}} \mathrm{d} u.
\end{align}
 The mathematical expectation is denoted by $\Ex{\cdot}$, and the notation $\mathcal{CN}\left( \eta,\sigma^2\right) $ represents the complex Gaussian distribution with mean $\eta$ and variance $\sigma^2$. 

The rest of the manuscript is organized as follows: Section~\ref{sec:sys} describes the system model and formulates the problem. The main analytical results along with several numerical investigations are then presented in Section~\ref{sec:main}. The dimensional trade-off between the transmitter and the \ac{irs} is investigated in Section~\ref{sec:trade}. Section~\ref{sec:derive} provides the detailed derivations of the main results. Finally, Section~\ref{sec:conc} concludes the manuscript.


\section{Problem Formulation}
\label{sec:sys}
Consider a \ac{miso} system in which a \ac{bs} with an array-antenna of size $M$ transmits data to a single-antenna user. An \ac{irs} with $N$ reflecting elements is further employed to modify the propagation environment between the \ac{bs} and the user: Each element of the \ac{irs} reflects its received signal after applying a phase-shift.~The~signal received by the user is then the superposition of two components: One that is received through the direct path between the \ac{bs} and the user, and one that is reflected by the \ac{irs}.

\subsection{System Model}
As suggested by the literature \cite{bjornson2020rayleigh,zhang2021reconfigurable}, we assume that the \ac{bs} and the \ac{irs} are equipped with planar arrays. Namely, we assume the antenna-array at the \ac{bs} and the array of reflecting elements on the \ac{irs} are rectangular planar arrays with $M_\xx$ and $N_\xx$ horizontal elements and $M_\yy$ and $N_\yy$ vertical elements, respectively, such that $M=M_\xx M_\yy$ and $N=N_\xx N_\yy$. The antenna elements at the \ac{bs} are assumed to be distanced with $\ell_\xx$ and $\ell_\yy$ on the horizontal and vertical axes, respectively. The horizontal and vertical distances at the \ac{irs} are further denoted by $d_\xx$ and $d_\yy$, respectively.

To model the direct path\footnote{Note that the direct path is different from the line of sight.}, we consider a classical scenario in which the \ac{los} link is blocked via obstacles or mobility of the user. It is further assumed that $\ell_\xx$ and $\ell_\yy$ are set large enough, such that the spatial correlation can be ignored\footnote{This is a rational assumption, as we do not assume $M$ to be large.}. We therefore adopt the Rayleigh fading model in which the channel coefficient of the direct path between antenna element $m$ at the \ac{bs} and the user is modeled as 
\begin{align}
	h_{\mathrm{d},m}=\sqrt{\alpha_\mathrm{d} A_M } \tilde{h}_{\mathrm{d},m}.
\end{align}
Here, $\alpha_\mathrm{d}$ is a distance-dependent path-loss, $A_M$ denotes the area of a single element on the planar antenna array~at~the \ac{bs}, i.e.,  $A_M = \ell_\xx \ell_\yy$, and $\tilde{h}_{\mathrm{d},m}$ follows a circularly-symmetric complex Gaussian distribution with zero mean and unit variance, i.e., $\tilde{h}_{\mathrm{d},m}\sim \mathcal{CN}\left(0, 1\right)$. For sake of brevity, we define the direct channel vector as $\mathbf{h}_\mathrm{d}=[h_{\mathrm{d},1}, \cdots, h_{\mathrm{d},M}]^\trp$.

Due to the flexibility in deploying the \ac{irs}, it is reasonable to assume that the link between the \ac{bs} and the \ac{irs} is dominated by a \ac{los} component. We denote the channel spanning from the \ac{bs} to the \ac{irs} with ${\mathbf{T}}\in \mathbb{C}^{N \times M}$  where $t_{nm}=[{\mathbf{T}}]_{nm}$
denotes the channel coefficient between the $m$-th transmit antenna  and the $n$-th \ac{irs} element. In particular, $t_{mn}$ is given by
\begin{align}
	t_{nm}=\sqrt{\alpha_{\mathrm{s}} A_N }\bar{t}_{nm},
\end{align}
where $\alpha_{\mathrm{s}}$ represents the distance-dependent path-loss, $A_N$ is the area of a single reflecting element, i.e., $A_N = d_\xx d_\yy$, and $\bar{t}_{nm}$ denotes the \ac{los} components.

Unlike the direct and the \ac{bs}-to-\ac{irs} links, the link between the user and the \ac{irs} has often both \ac{los} and \ac{nlos} components. Moreover, the reflecting elements at the \ac{irs} are often closely distanced, such that the given area of the \ac{irs} is filled with a large number of elements. As a result, the spatial correlation in this channel cannot be ignored. We hence use the Rician fading model with spatial correlation to model this channel: Let $h_{\mathrm{r}, n}$ denote the coefficient of the channel between the $n$-th reflecting element and the user. This coefficient is modeled as
\begin{align} \label{eq:hr} 
	h_{\mathrm{r}, n}=\sqrt{\alpha_{\mathrm{ r}} A_N }\left(\sqrt{\frac{\kappa_\mathrm{r}}{\kappa_\mathrm{r}+1}}\bar{h}_{\mathrm{r}, n} +\sqrt{\frac{1}{\kappa_\mathrm{r}+1}}\tilde{h}_{\mathrm{r}, n}\right), 
\end{align}
where $\alpha_{\mathrm{ r}}$ and $\kappa_\mathrm{r}$ are  the distance-dependent path-loss and the Rician factor\footnote{In general,  larger $\kappa_\mathrm{r}$ means that the channel is more deterministic.}, respectively. The coefficient $\bar{h}_{\mathrm{r}, n}$ represents~the \ac{los}, and $\tilde{h}_{\mathrm{r}, n}$ models the small-scale fading process in the \ac{nlos} link. To capture the spatial correlation among the \ac{irs} elements, we consider a general covariance matrix $\mR\in \setC^{N\times N}$ for the vector $\tilde{\mathbf{h}}_{\mathrm{r}}=[\tilde{h}_{\mathrm{r}, 1}, \cdots, \tilde{h}_{\mathrm{r}, N}]^\trp$ which depends on $d_\xx$ and $d_\yy$, distribution of scatterers, and the radiation pattern. For the case with isotropic scattering in front of the \ac{irs}, one can determine $\mR$ explicitly from \cite[Proposition~1]{bjornson2020rayleigh}. Consequently, $\tilde{\mathbf{h}}_{\mathrm{r}}$ is modeled as a zero-mean complex Gaussian vector with covariance $\mR$, i.e., $\tilde{\mathbf{h}}_{\mathrm{r}} \sim \mathcal{CN}\left(\boldsymbol{0}, \mR\right)$. Note that due to power normalization, we assume that $[\mR]_{nn} =1$ for $n\in[N]$. 
 

\begin{remark}
	In general, the effective area of a single element on an array depends on the wave-length, and hence the given expressions for the array areas, i.e., $A_N = d_\xx d_\yy$ and $A_M=\ell_\xx \ell_\yy$, are not exact if the elements are distances longer than a half wave-length. Nevertheless, we assume the neighboring elements on the transmit array and the \ac{irs} to be distanced less than a half wave-length, i.e., $d_\xx,d_\yy,\ell_\xx,\ell_\yy \leq \lambda/2$.
\end{remark}

The \ac{los} components are described by the array responses: Let $\lambda$ be the wavelength and define the following operators 
\begin{subequations}
	\begin{align}
		i_M\brc{m} &=\brc{m-1} \;  \mathrm{mod} \; M_\xx,  \\
		i_N\brc{n} &= \brc{n-1}  \;  \mathrm{mod} \;  N_\xx, 
\end{align}
\end{subequations}
where $x \;  \mathrm{mod} \;  L$ determines $x$ modulo $L$. Moreover, let
\begin{subequations}
	\begin{align}
j_M\brc{m} &= \left\lfloor \displaystyle \frac{m-1}{M_\xx} \right\rfloor , \\
j_N\brc{n} &= \left\lfloor \displaystyle \frac{n-1}{N_\xx} \right\rfloor.
\end{align}
\end{subequations}
 We now define an exponent function at azimuth angle $\varphi$ and elevation angle $\theta$ for a given element $m$ on the \ac{bs} array and element $n$ on the \ac{irs} array, respectively, as follows:
\begin{subequations}
	\begin{align}
		\Phi_m\left( \varphi,\theta\right) &= {i_M\brc{m} \ell_\xx \cos{\theta} \sin{\varphi}  +  j_M\brc{m} \ell_\yy \sin{\theta} },\\
		\Pi_n\left( \varphi,\theta\right) &={i_N\brc{n} d_\xx \cos{\theta} \sin{\varphi} +  j_N\brc{n} d_\yy \sin{\theta} },
	\end{align}
\end{subequations}
Consequently, the transmit and \ac{irs} array responses are given at $\brc{\varphi,\theta}$ respectively by \cite{bjornson2020rayleigh}
\begin{subequations}
	\begin{align}
		\mathbf{a}_M\left( \varphi,\theta\right) &= \left[ \mathrm{e}^{   \tfrac{2\pi\jj }{\lambda} \Phi_1\left( \varphi,\theta\right)  }, \cdots, \mathrm{e}^{  \tfrac{2\pi\jj }{\lambda} \Phi_M\left( \varphi,\theta\right) } \right]^\trp,\\
		\mathbf{a}_N\left( \varphi, \theta\right) &= \left[ \mathrm{e}^{ \tfrac{2\pi\jj }{\lambda}  \Pi_1\left( \varphi,\theta\right)  }, \cdots, \mathrm{e}^{  \tfrac{2\pi\jj }{\lambda}  \Pi_N\left( \varphi,\theta\right) } \right]^\trp.
	\end{align}
\end{subequations}

Let $\bar{\mathbf{T}} \in \setC^{N\times M}$ be the matrix of \ac{los} channel coefficients between the \ac{bs} and the \ac{irs} whose entry $\brc{n,m}$ is $\bar{t}_{nm}$. Define further the \ac{los} channel vector between the \ac{irs} and the user as $\mathbf{\bar{h}}_{\mathrm{r}}=[\bar{h}_{\mathrm{r}, 1}, \cdots, \bar{h}_{\mathrm{r}, N}]^\trp$. We can hence write
\begin{subequations}
\begin{align}
	\bar{\mathbf{T}} &= \mathbf{a}_N\left( \varphi_{\mathrm{r}1},\theta_{\mathrm{r}1}\right) \mathbf{a}_M\left( \varphi_{\mathrm{t}2} , \theta_{\mathrm{t}2} \right)^\mathsf{H}, \label{a1}\\
	\mathbf{\bar{h}}_{\mathrm{r}} &= \mathbf{a}_N\left( \varphi_{\mathrm{t}1} , \theta_{\mathrm{t}1} \right), \label{a2}
\end{align}
\end{subequations}
where $\brc{\varphi_{\mathrm{r}1},\theta_{\mathrm{r}1}}$ is the \ac{aoa} at the \ac{irs}, the pair $\brc{\varphi_{\mathrm{t}1} , \theta_{\mathrm{t}1} }$ denotes the \ac{aod} from the \ac{irs}, and $\brc{\varphi_{\mathrm{t}2} , \theta_{\mathrm{t}2} }$ is the \ac{aod} from the \ac{bs}.
	
The signal received by the user is given by
\begin{align}\label{eq:rx}
	y=\sum_{m=1}^{M} h_m x_m + z,
\end{align}
where $z$ is zero-mean and unit-variance complex Gaussian noise, and $x_m$ denotes the symbol sent by the $m$-th antenna element at the \ac{bs}. The transmitted symbols satisfy
\begin{align}
	 \sum_{m=1}^{M} \Ex{\abs{x_m}^2} \leq \rho,
\end{align}
for some total transmit power $\rho$. The coefficient $h_m$ in \eqref{eq:rx} defines for the \textit{end-to-end} effective channel between the $m$-th \ac{bs} antenna and the user which is given by
\begin{align}
	h_m
	&=h_{\mathrm{d}, m}+\sum_{n=1}^{N} \mathrm{e}^{-\jj\beta_n} h_{\mathrm{r}, n} t_{nm}, 
\end{align}
where $\beta_n$ denotes the phase-shift with element $n$ at the \ac{irs}.
%
%
%


\subsection{Asymptotic Scaling of IRS}
We intend to characterize the \ac{irs}-aided setting in a large-system limit in which the transmit-array size $M$ remains fixed and the number of reflecting elements at the \ac{irs} grows large. In general, the growth in the number of reflecting elements can impact the physical dimension of the \ac{irs} in various forms. To illustrate this point, let $A_{\rm IRS}$ denote the total area of the \ac{irs}. Assuming the reflecting elements to be symmetrically inserted on the surface, we can write
\begin{align}
	A_{\rm IRS} =  N A_N.
\end{align}

The scaling of $A_{\rm IRS}$ in terms of $N$ is illustrated at best by considering the two most extreme scenarios.
\begin{itemize}
	\item One extreme case is to assume that the number of \ac{irs} elements grows large and the distances between~neighboring elements on the \ac{irs} are kept fixed and large enough, e.g., $d_\xx = d_\yy = \lambda/2$. In this case, $A_N$ is fixed,~and~therefore the total area of the \ac{irs} grows linearly in $N$. Depending on the wave-length, this scaling can lead to unrealistic physical dimensions of \ac{irs} for large choices of $N$. 
	\item Another extreme case is when we set the total area of the \ac{irs} fixed, and reduce the distance between the neighboring elements by the growth of $N$. In this case. $A_{\rm IRS}$ does not scale in $N$, and $A_N$ shrinks reverse-linearly with $N$. In practice, however, the physical dimensions of a typical reflecting element cannot be set below a certain threshold. Hence, for a very large choice of $N$, assuming this extreme case can also be unrealistic.
\end{itemize}

In practice, with larger numbers of elements on the~\ac{irs},~the distance among neighboring elements is reduced. This leads to a smaller effective area for a reflecting element and therewith to a higher spatial correlation. The distance can however be reduced down to a certain level. We can hence conclude that the scaling of the \ac{irs} area is in general neither linear nor fixed, but of some intermediate order; see some related discussions in \cite{ivrlac2010toward,ivrlac2014multiport}. To take this point into account, we consider a basic scaling model for the \ac{irs} area. Namely, we assume that the area of each \ac{irs} element scales with $N$ as %
\begin{align}
A_N = \frac{A_0}{N^{q}}
\end{align}
for some constant $A_0$ and $0\leq q \leq 1$. Consequently, the total area of the \ac{irs} scales with $N$ as
\begin{align}
	A_{\rm IRS} =  A_0 N^{1-q}.
\end{align}
This basic scaling addresses the limiting scenarios in between the two extreme cases.
\begin{itemize}
	\item $q=0$ is the idealistic case which corresponds to the first extreme case, i.e., the case in which the distances between neighboring elements are kept fixed.
	\item $q=1$ addresses the latter case in which the total area of the \ac{irs} is fixed.
\end{itemize}
By setting $0< q < 1$, an intermediate scaling is considered in which the \ac{irs} surface grows sub-linearly by $N$.

It is worth mentioning that for $q\neq 1$, this scaling model implies that by sending $N\rightarrow \infty$, the area of \ac{irs} also grows asymptotically large. This observation can lead to this conclusion that with $q\neq 1$, the far-field derivation for the \ac{irs} array response is no longer valid in the large-system limit. To avoid such inconsistency, let $d_{\min}$ denotes the minimum distance between two terminals in the network, i.e., the minimum of the \ac{bs}-to-user, \ac{bs}-to-\ac{irs} and \ac{irs}-to-user distances. We assume that $d_{\min}$ is bounded uniformly from below as
\begin{align}
	d_{\min} > D_0N^{{\gamma}/{2}},
\end{align}
for some $D_0$ and $\gamma > 1-q$. By considering this assumption, we guarantee that the far-field derivations are valid through the asymptotic analyses. From the implementational point of view, this is a realistic assumption, as the distances are often in orders of tens or hundreds of meters while the \ac{irs} dimensions are often a finite multiple of the wave-length.

\subsection{Input-Output Mutual Information}
Let $\mathbf{h}=[h_1, \cdots, h_M]^\trp$ denote the end-to-end channel vector. Considering Gaussian signaling, the mutual information between the input and output of the end-to-end channel, for a given realization of $\mathbf{h}$, is given by
\begin{align}
	\mathcal{I}\left( \mathbf{h}, \rho\mathbf{Q}\right) =\log_2 \left(1+ \rho \mathbf{h}^\mathsf{H} \mathbf{Q}\mathbf{h}\right), 
\end{align}
where $\mathbf{Q} \in \setC^{M \times M}$ is the transmit covariance matrix. From the power constraint, we know that $\mathbf{Q}$ satisfies $\tr{\mathbf{Q}} = 1$. 

In this channel, the optimal covariance matrix which maximizes the mutual information is given by \ac{mrt} \cite{tse2005fundamentals,loyka2017capacity}. We hence focus on the optimal case which determines the capacity of the end-to-end channel. To this end, we define the maximum mutual information\footnote{We show it by $\mathcal{C}$, as it determines the end-to-end \textit{capacity}.} as
\begin{align}\label{eq:I}
	\mathcal{C} =\log_2 \left(1+ \rho M \Gamma \right),
\end{align}
which determines the end-to-end channel capacity. Here, $\Gamma$ is given by
\begin{align}
 \Gamma = \frac{ \norm{ \mathbf{h} }^2}{M}.
\end{align}
In the sequel, we refer to $\Gamma$ as the \textit{end-to-end \ac{snr} gain}~per transmit antenna. Due to fading, $\Gamma$ and $\mathcal{C}$ are random. Our main goal is to characterize the statistics of these random variables analytically, when $N$ is asymptotically large.

\section{Asymptotic Channel Hardening}
\label{sec:main}
Intuitively, channel hardening in massive \ac{mimo} systems refers to the following phenomenon: A fading \ac{mimo} channel behaves almost deterministically, as the number of antennas at one end grows large\footnote{In general, it is enough to have a large antenna array at one end.}. This fundamental property is characterized~in~the seminal work \cite{hochwald2004multiple} in which a tight approximation for the distribution of the input-output mutual information is determined. The result shows that the variance of the mutual information shrinks rapidly while its mean grows large. 

Invoking the literature, it can be straightforwardly shown that the end-to-end channel $\mathbf{h}$ hardens, as the number of transmit antennas $M$ grows unboundedly large. We are however interested in a different asymptotic regime; namely, a scenario with \textit{finite} transmit antennas but unboundedly \textit{large number of \ac{irs} elements}, i.e., fixed $M$ and $N\rightarrow\infty$. Our main result shows that in this alternative asymptotic regime, the channel still hardens, under a very mild constraint on the correlation among the \ac{irs} elements.

\begin{proposition}
	\label{th:0}
	Let the area of the \ac{irs} scale with $N$ as $A_{\rm IRS} =  A_0 N^{1-q}$ for some fixed $A_0$ and $0\leq q <1$, and the phase-shifts be set to 
	\begin{align} 
	\beta_n^\star= \frac{2\pi }{\lambda} \brc{\Pi_n\left( \varphi_{\mathrm{r}1} ,\theta_{\mathrm{r}1} \right)  + \Pi_n\left( \varphi_{\mathrm{t}1} ,\theta_{\mathrm{t}1} \right)}.  \label{eq:beta_n_star}
\end{align}
	Assume the maximum eigenvalue of the \ac{irs} covariance matrix $\mR$, denoted by $\lambda_{\max}$, grows with $N$ sub-linearly, i.e.,
	\begin{align}
		\lim_{N\rightarrow\infty} \frac{\lambda_{\max}}{N} = 0,
	\end{align}
 and satisfies
		\begin{align}
		\lim_{N\rightarrow\infty} \frac{1}{\lambda_{\max} A_{\rm IRS}} = 0. \label{eq:const_IRS}
	\end{align}
	 Then, the maximum mutual information $\mathcal{C}$ is well approximated by a real Gaussian random variable whose mean and~standard~deviation are given by
	\begin{subequations}
		\begin{align}
			\mu_{\mathcal{C}} &=   \log_2 \left(1 +\rho M \mu \right) \\
			\sigma_{\mathcal{C}} &= \frac{ \displaystyle \rho M  \log_2 \mathrm{e} }{1+ \rho M \mu }  \sqrt{ \omega \eta + \eta + \frac{M-1}{M} \alpha_{\mathrm{ d}} A_M }
		\end{align}
	\end{subequations}
	respectively, for
\begin{subequations}
		\begin{align}
		\mu &=  \alpha_{\mathrm{ d}} A_M+ \kappa_{\mathrm{ r} } \bar{\alpha}_N N^2 + \bar{\alpha}_N \mathbf{\bar{h}}_{\mathrm{r}}^\her \mR \mathbf{\bar{h}}_{\mathrm{r}} ,\\
		\eta &=  \frac{\alpha_{\mathrm{ d}}A_M}{M} +  \bar{\alpha}_N \mathbf{\bar{h}}_{\mathrm{r}}^\her \mR \mathbf{\bar{h}}_{\mathrm{r}} ,\\
		\omega &=  2\kappa_{\mathrm{ r} } \bar{\alpha}_N N^2 + \bar{\alpha}_N \mathbf{\bar{h}}_{\mathrm{r}}^\her \mR \mathbf{\bar{h}}_{\mathrm{r}} ,
	\end{align}
\end{subequations}
	 and $\bar{\alpha}_N$ being defined as
	\begin{align}
	\bar{\alpha}_N=  \frac{\alpha_{\mathrm{ r}} \alpha_{\mathrm{ s}} }{1+\kappa_\mathrm{r}} A_N^2  . \label{eq:bar_alpha}
	\end{align}
\end{proposition}

\begin{proof}
The proof is given in Section~\ref{proof:Prop}.
\end{proof}

Proposition~\ref{th:0} considers three main constraints on the system setting; namely, it restricts the scaling order of the \ac{irs} area, assumes sub-linearly growing $\lambda_{\max}$ and the constraint in \eqref{eq:const_IRS}. In this respect, it is worth mentioning few remarks.
\begin{itemize}
	\item 
	Proposition~\ref{th:0} excludes the limiting case of $q=1$, for analytical rigor. Nevertheless, as we see in the forthcoming section, the analytical expressions are still rather accurate for $q=1$.
	\item Constraining the growth rate of the maximum eigenvalue $\lambda_{\max}$ to be sub-linear can be interpreted as follows: $\lambda_{\max}$ is uniformly bounded from above as $\lambda_{\max} \leq aN^u$ for some real scalars $a$ and $0\leq u <1$. This is in general not a strong constraint, as we can write
	\begin{align}
		\lambda_{\max} \leq \tr{\mR} = N.
	\end{align}
As we illustrate later on, the exponent $u$~is~mutually~coupled with $q$. This is seen by considering the two extreme cases of $q$. For $q=0$, it is feasible to have $\mR=\mI_N$, and hence $\lambda_{\max} = 1$ meaning that $u=0$. The other extreme case is $q \rightarrow 1$ in which $\mR$ tends to a rank-one matrix. In this case $\lambda_{\max} = N$, and hence the uniform upper-bound is valid with\footnote{Similar to $q$, we exclude the linear case, i.e., $u=1$ for analytical rigor.} $a = u =1$. For less rank-deficient covariance matrices,  $u$ can be something in $\dbc{0,1}$. Similar to the constraint on the physical dimension of the \ac{irs}, we show later that despite assuming $0\leq u <1$, Proposition~\ref{th:0} is still valid for the extreme case of $u=1$. 
\item  Considering the scaling order of the \ac{irs} area $A_{\rm IRS}$ and maximum eigenvalue $\lambda_{\max}$, the constraint in \eqref{eq:const_IRS} restricts $u$ and $q$ to satisfy $u\geq q$.  One may find this constraint intuitively valid following from the above discussions on the coupling of $u$ and $q$. We later on show that this is in fact the case for the covariance matrix derived for the standard Rayleigh model in \cite{bjornson2020rayleigh}.
\end{itemize}

\subsection{Numerical Investigations}
\label{sec:Numerical}
Before starting with derivations, let us confirm the accuracy of Proposition~\ref{th:0} for practical system dimensions through~some numerical experiments. To this end, we consider a basic scenario in which the \ac{bs} is equipped with a $2 \times 2$ planar array and the \ac{irs} surface consists of $N=256$ reflecting elements which are aligned on a rectangular surface with $N_\xx = 8$ horizontal elements and $N_\yy = 32$ vertical elements. The elements at both \ac{bs} and \ac{irs} are distanced with $\ell_\xx = d_\xx = \ell_\yy = d_\yy = \lambda /2$.

For sake of brevity, we set $\alpha_{\mathrm{ d}} A_M = \alpha_{\mathrm{ r}} A_N= \alpha_{\mathrm{ s}} A_N=1$, and $\log \kappa_{\mathrm{ r} } = 0$ dB. For the covariance matrix $\mR$, we consider a Rayleigh fading model and invoke \cite[Proposition~1]{bjornson2020rayleigh}. The entry $\brc{n,n'}$ of $\mR$ is hence given by
\begin{align}
	\dbc{ \mR }_{n n'} = \mathrm{sinc} \brc{\frac{2}{\lambda} \sqrt{D_\xx^2 + D_\yy^2 } } \label{eq:RR}
\end{align}
with $D_\xx$ and $D_\yy$ being
\begin{subequations}
	\begin{align}
	D_\xx &= d_\xx \dbc{i_N\brc{n} - i_N\brc{n'}},\\
	D_\yy &= d_\yy \dbc{j_N\brc{n} - j_N\brc{n'}}.
	\end{align}
\end{subequations}
The \ac{aoa} and \acp{aod} are further set to $\brc{\varphi_{\mathrm{r}1},\theta_{\mathrm{r}1}} = \brc{\pi/6,\pi/3 }$, $\brc{\varphi_{\mathrm{t}1},\theta_{\mathrm{t}1}} = \brc{\pi/8,2\pi/3 }$ and $\brc{\varphi_{\mathrm{t}2},\theta_{\mathrm{t}2}} = \brc{\pi/7,\pi/5 }$. The power constraint is set to $\rho=1$. 

For this setting, we collect $10^5$ realizations of the channel and determine the input-output mutual information for every realization. The histogram of the collected samples is shown in Fig.~\ref{fig:hist}. As the figure shows, the histogram closely matches the Gaussian distribution. We now plot the analytical distribution given by Proposition~\ref{th:0} and compare it to the Gaussian distribution fitted to the histogram. The result is shown in Fig.~\ref{fig:dist}. As the figure demonstrates, the analytical result of Proposition~\ref{th:0} matches almost perfectly to the empirical density. 

 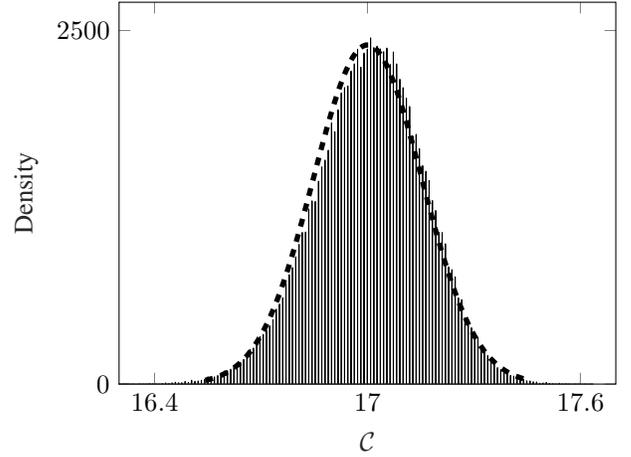
\begin{figure}
	\centering
%
%
\definecolor{mycolor1}{rgb}{0.00000,0.44700,0.74100}%
\begin{tikzpicture}

\begin{axis}[%
width=2.6in,
height=2in,
at={(1.262in,0.697in)},
scale only axis,
bar shift auto,
xmin=16.3,
xmax=17.7,
xtick={16.4,17,17.6},
xticklabels={{$16.4$},{$17$},{$17.6$}},
xlabel style={font=\color{white!15!black}},
xlabel={$\mathcal{C}$},
ymin=0,
ymax=2700,
ytick={0,2500},
yticklabels={{$0$},{$2500$}},
ylabel style={font=\color{white!15!black}},
ylabel={Density},
axis background/.style={fill=white},
legend style={legend cell align=left, align=left, draw=white!15!black}
]
\addplot[ybar, bar width=0.01, fill=mycolor1, draw=black, area legend,forget plot] table[row sep=crcr] {%
16.258585	3\\
16.267755	0\\
16.276925	0\\
16.286095	0\\
16.295265	0\\
16.304435	1\\
16.313605	1\\
16.322775	0\\
16.331945	1\\
16.341115	0\\
16.350285	2\\
16.359455	1\\
16.368625	1\\
16.377795	4\\
16.386965	2\\
16.396135	2\\
16.405305	1\\
16.414475	4\\
16.423645	7\\
16.432815	4\\
16.441985	8\\
16.451155	12\\
16.460325	9\\
16.469495	9\\
16.478665	18\\
16.487835	11\\
16.497005	26\\
16.506175	19\\
16.515345	21\\
16.524515	26\\
16.533685	32\\
16.542855	51\\
16.552025	40\\
16.561195	50\\
16.570365	60\\
16.579535	72\\
16.588705	90\\
16.597875	94\\
16.607045	104\\
16.616215	126\\
16.625385	130\\
16.634555	153\\
16.643725	174\\
16.652895	215\\
16.662065	243\\
16.671235	254\\
16.680405	283\\
16.689575	335\\
16.698745	350\\
16.707915	379\\
16.717085	411\\
16.726255	455\\
16.735425	529\\
16.744595	565\\
16.753765	609\\
16.762935	712\\
16.772105	772\\
16.781275	823\\
16.790445	898\\
16.799615	987\\
16.808785	1072\\
16.817955	1074\\
16.827125	1237\\
16.836295	1295\\
16.845465	1291\\
16.854635	1434\\
16.863805	1535\\
16.872975	1583\\
16.882145	1653\\
16.891315	1846\\
16.900485	1783\\
16.909655	1939\\
16.918825	2058\\
16.927995	2096\\
16.937165	2107\\
16.946335	2213\\
16.955505	2263\\
16.964675	2367\\
16.973845	2240\\
16.983015	2340\\
16.992185	2368\\
17.001355	2447\\
17.010525	2379\\
17.019695	2388\\
17.028865	2369\\
17.038035	2350\\
17.047205	2375\\
17.056375	2260\\
17.065545	2347\\
17.074715	2263\\
17.083885	2155\\
17.093055	2095\\
17.102225	2021\\
17.111395	1962\\
17.120565	1751\\
17.129735	1765\\
17.138905	1664\\
17.148075	1544\\
17.157245	1504\\
17.166415	1439\\
17.175585	1297\\
17.184755	1227\\
17.193925	1072\\
17.203095	1071\\
17.212265	991\\
17.221435	828\\
17.230605	808\\
17.239775	757\\
17.248945	609\\
17.258115	596\\
17.267285	468\\
17.276455	453\\
17.285625	398\\
17.294795	333\\
17.303965	328\\
17.313135	258\\
17.322305	238\\
17.331475	201\\
17.340645	157\\
17.349815	144\\
17.358985	116\\
17.368155	101\\
17.377325	76\\
17.386495	77\\
17.395665	61\\
17.404835	55\\
17.414005	36\\
17.423175	33\\
17.432345	32\\
17.441515	22\\
17.450685	21\\
17.459855	16\\
17.469025	14\\
17.478195	5\\
17.487365	7\\
17.496535	11\\
17.505705	4\\
17.514875	3\\
17.524045	1\\
17.533215	5\\
17.542385	1\\
17.551555	3\\
17.560725	1\\
17.569895	0\\
17.579065	1\\
17.588235	1\\
17.597405	0\\
17.606575	1\\
17.615745	0\\
17.624915	1\\
};
\addplot[forget plot, color=white!15!black] table[row sep=crcr] {%
16.246664	0\\
17.636836	0\\
};

\addplot [color=black,dashed, line width=2.0pt,forget plot]
  table[row sep=crcr]{%
16.5427161668143	26.6462485983053\\
16.5519595954918	31.900742534959\\
16.5612030241692	38.0513716365661\\
16.5704464528466	45.2214671710287\\
16.5796898815241	53.5455998872424\\
16.5889333102015	63.1695431651569\\
16.5981767388789	74.2500057425097\\
16.6074201675563	86.9540980782759\\
16.6166635962338	101.458497309248\\
16.6259070249112	117.948277926871\\
16.6351504535886	136.615378894518\\
16.6443938822661	157.656683058989\\
16.6536373109435	181.271691460231\\
16.6628807396209	207.659783536915\\
16.6721241682983	237.01706422694\\
16.6813675969758	269.53281046905\\
16.6906110256532	305.385542444924\\
16.6998544543306	344.738758804846\\
16.7090978830081	387.736389757863\\
16.7183413116855	434.498036867103\\
16.7275847403629	485.114083193294\\
16.7368281690403	539.640771534183\\
16.7460715977178	598.095361336204\\
16.7553150263952	660.451485803024\\
16.7645584550726	726.634839190525\\
16.7738018837501	796.519329683825\\
16.7830453124275	869.923835072333\\
16.7922887411049	946.609696233829\\
16.8015321697823	1026.27907687091\\
16.8107755984598	1108.57430681523\\
16.8200190271372	1193.07831047738\\
16.8292624558146	1279.31620180253\\
16.8385058844921	1366.75810269943\\
16.8477493131695	1454.82321383539\\
16.8569927418469	1542.8851356052\\
16.8662361705243	1630.27840382107\\
16.8754795992018	1716.30617021398\\
16.8847230278792	1800.24892327505\\
16.8939664565566	1881.37411145992\\
16.9032098852341	1958.94649953014\\
16.9124533139115	2032.23906098384\\
16.9216967425889	2100.54418625153\\
16.9309401712663	2163.18496859281\\
16.9401835999438	2219.52631826154\\
16.9494270286212	2268.98565114218\\
16.9586704572986	2311.04290108994\\
16.9679138859761	2345.24961576813\\
16.9771573146535	2371.23691372251\\
16.9864007433309	2388.72210534772\\
16.9956441720083	2397.51381159857\\
17.0048876006858	2397.51545084956\\
17.0141310293632	2388.72700507059\\
17.0233744580406	2371.24502015708\\
17.0326178867181	2345.26084040607\\
17.0418613153955	2311.0571222833\\
17.0511047440729	2269.00271629792\\
17.0603481727503	2219.54604656638\\
17.0695916014278	2163.20715420059\\
17.0788350301052	2100.56860185457\\
17.0880784587826	2032.26546168157\\
17.0973218874601	1958.97462690338\\
17.1065653161375	1881.40369775625\\
17.1158087448149	1800.27969561124\\
17.1250521734923	1716.33785470427\\
17.1342956021698	1630.31072954557\\
17.1435390308472	1542.91783834054\\
17.1527824595246	1454.85603947899\\
17.1620258882021	1366.7908103302\\
17.1712693168795	1279.34856633262\\
17.1805127455569	1193.11012486266\\
17.1897561742343	1108.60538380474\\
17.1989996029118	1026.30925028848\\
17.2082430315892	946.63882179702\\
17.2174864602666	869.951790764821\\
17.2267298889441	796.546015707253\\
17.2359733176215	726.660177530156\\
17.2452167462989	660.47541944821\\
17.2544601749763	598.117853201706\\
17.2637036036538	539.661803131665\\
17.2729470323312	485.133653103924\\
17.2821904610086	434.516159066249\\
17.2914338896861	387.753091841481\\
17.3006773183635	344.754080164566\\
17.3099207470409	305.39953243516\\
17.3191641757183	269.545526605028\\
17.3284076043958	237.028570447454\\
17.3376510330732	207.670148560239\\
17.3468944617506	181.280987256013\\
17.3561378904281	157.66498345343\\
17.3653813191055	136.622758320819\\
17.3746247477829	117.954810327888\\
17.3838681764603	101.464255194816\\
17.3931116051378	86.9591517355791\\
17.4023550338152	74.254422594533\\
17.4115984624926	63.1733872693006\\
17.4208418911701	53.5489315641595\\
17.4300853198475	45.2243427519477\\
17.4393287485249	38.0538433164026\\
17.4485721772023	31.9028583179244\\
17.4578156058798	26.6480523226081\\
};

\end{axis}
\end{tikzpicture}%
	\caption{Histogram of the mutual information and the fitted Gaussian distribution for the covariance matrix $\mR$ given in \eqref{eq:RR}.}
	\label{fig:hist}
\end{figure}

 \begin{figure}
	\centering
	\input{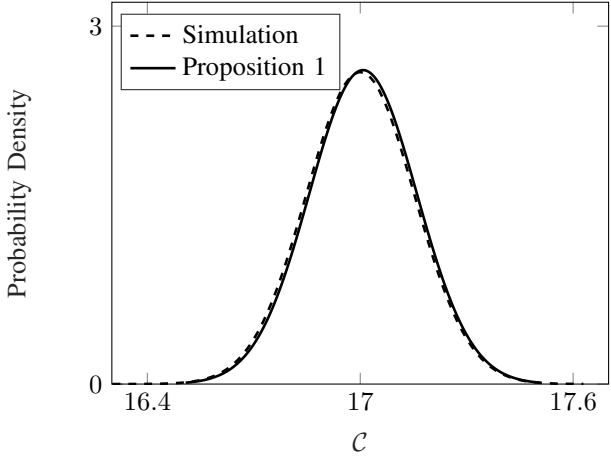}
	\caption{Probability density of the maximum mutual information and the Gaussian distribution fitted to the numerical simulations.}
	\label{fig:dist}
\end{figure}

We now conduct a new experiment. We consider the same setting and let the \ac{irs} size to vary as $N = N_\xx^2$ while changing from $N_\xx = 8$ to $N_\xx = 36$, gradually. Note that in this setting, the area of the \ac{irs} grows linearly\footnote{Such scaling is not realistic for very large \acp{irs}. We discuss more realistic scaling scenarios in the forthcoming parts of this section.} in $N$, i.e., $t = 0$, and hence the limit in \eqref{eq:const_IRS} is satisfied. For each choice of $N$, we collect $500$ samples, determine numerically the mean and variance of $\mathcal{C}$ and plot it against $N$ in Figs.~\ref{fig:mean} and \ref{fig:var}. 

The numerical results in Figs.~\ref{fig:mean} and \ref{fig:var} are compared with the closed-form expressions given in Proposition~\ref{th:0}. As the figures show, the analytical expressions closely track the numerical results. The figures further demonstrate the hardening of the channel in terms of the \ac{irs} size. As $N$ grows large, the mean mutual information grow large, while the variance drops.

\subsection{Channel Hardening Order}
\label{sec:IRS_Size}
The numerical experiments in the previous section consider the case of $q=0$, i.e., when the \ac{irs} area grows linearly in $N$. Although this can be considered feasible for a small or moderate number of reflecting elements, it is not a realistic scaling for asymptotically large \acp{irs}. In fact, in such scenarios, the total area of the \ac{irs} is limited, and hence, the area grows sub-linearly, i.e., $q > 0$. We address this point by investigating the speed of channel hardening in this section. %
We are mainly interested to find out how fast the channel hardens with respect to $N$ for a given scaling of the \ac{irs} area. In general, the speed depends on the correlation among the reflecting elements, and therefore the physical dimensions of the \ac{irs}. This argument is analytically characterized in Proposition~\ref{th:0-1}.

 \begin{figure}
	\centering
%
%
\begin{tikzpicture}

\begin{axis}[%
	width=2.6in,
	height=2in,
	at={(1.262in,0.697in)},
	scale only axis,
	xmin=30,
	xmax=1390,
	xtick={100,500,900,1300},
	xticklabels={{$100$},{$500$},{$900$},{$1300$}},
	xlabel style={font=\color{white!15!black}},
	xlabel={$N$},
	ymin=12,
	ymax=22,
	ytick={15,20},
	yticklabels={{$15$},{$20$}},
	ylabel style={font=\color{white!15!black}},
	ylabel={$\mu_{\mathcal{C}}$},
	axis background/.style={fill=white},
	legend style={at={(.97,.23)},legend cell align=left, align=left, draw=white!15!black}
]
\addplot [color=black, draw=none, mark size=3.0pt, mark=o, mark options={solid, black}]
  table[row sep=crcr]{%
64	12.9970684902932\\
100	14.2883841659514\\
144	15.3292297021062\\
196	16.2336621379149\\
256	17.0070276728797\\
324	17.6865847711985\\
400	18.2865861473418\\
484	18.8363078895447\\
576	19.3397787785154\\
676	19.7998357227501\\
784	20.2263418543756\\
900	20.6297261075209\\
1024	21.0012587825713\\
1156	21.3494913165317\\
1296	21.6751805396198\\
};
\addlegendentry{Simulation}

\addplot [color=black, dashed, line width=1.0pt]
  table[row sep=crcr]{%
64	13.0355483121664\\
100	14.3100808215705\\
144	15.3551023458256\\
196	16.2404574858361\\
256	17.0083288692013\\
324	17.6861783237922\\
400	18.2928832932337\\
484	18.8419657322054\\
576	19.3434189686008\\
676	19.8048269394158\\
784	20.2320871058475\\
900	20.6298876957136\\
1024	21.0020261819034\\
1156	21.3516220639022\\
1296	21.681262280042\\
};
\addlegendentry{Proposition 1}

\end{axis}
\end{tikzpicture}%
	\caption{Numerical simulations and the analytic expression for the mean mutual information against $N$.}
	\label{fig:mean}
\end{figure}
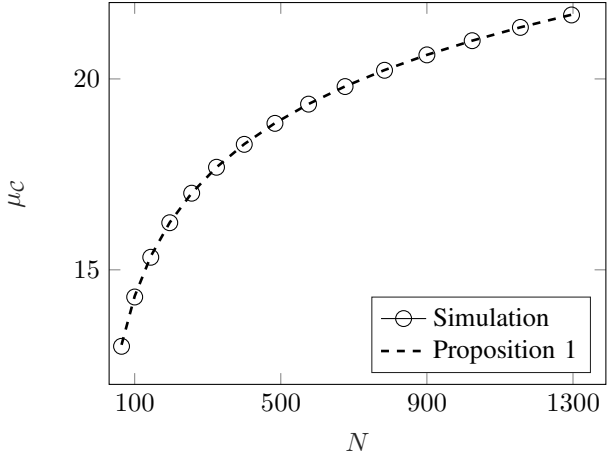

\begin{figure}
	\centering
%
%
\begin{tikzpicture}

\begin{axis}[%
	width=2.6in,
height=2in,
at={(1.262in,0.7in)},
scale only axis,
xmin=30,
xmax=1390,
xtick={100,500,900,1300},
xticklabels={{$100$},{$500$},{$900$},{$1300$}},
xlabel style={font=\color{white!15!black}},
xlabel={$N$},
ymode=log,
ymin=0.0035,
ymax=0.125,
ytick={.1,.01},
yticklabels={{$10^{-1}$},{$10^{-2}$}},
ylabel style={font=\color{white!15!black}},
ylabel={$\sigma^2_{\mathcal{C}}$},
yminorticks=true,
axis background/.style={fill=white},
legend style={at={(.97,.97)},legend cell align=left, align=left, draw=white!15!black}
]
\addplot [color=black, draw=none, mark size=3.0pt, mark=o, mark options={solid, black}]
  table[row sep=crcr]{%
64	0.110756033617599\\
100	0.0665236171208807\\
144	0.0432304836381895\\
196	0.0319728420295954\\
256	0.0234246428990208\\
324	0.0196940608164752\\
400	0.0137706619243541\\
484	0.0131669896718542\\
576	0.0102679255206928\\
676	0.00878947232805994\\
784	0.00822424044817825\\
900	0.00660170003380224\\
1024	0.00573928232825865\\
1156	0.0053249983405303\\
1296	0.00402217804319547\\
};
\addlegendentry{Simulation}

\addplot [color=black, dashed, line width=1.0pt]
  table[row sep=crcr]{%
64	0.0981198494945909\\
100	0.0627363215526332\\
144	0.0431512624661583\\
196	0.0313914768963716\\
256	0.0237677980761486\\
324	0.018529968430328\\
400	0.0148149757426382\\
484	0.012160643034831\\
576	0.0102474168118813\\
676	0.00881560485437426\\
784	0.00766881709886843\\
900	0.00668837771857914\\
1024	0.00583019755511801\\
1156	0.00509667185676224\\
1296	0.00449794341562863\\
};
\addlegendentry{Proposition~1}

\end{axis}
\end{tikzpicture}%
	\caption{Numerical simulations and the analytic expression for the variance of the mutual information against $N$.}
	\label{fig:var}
\end{figure}
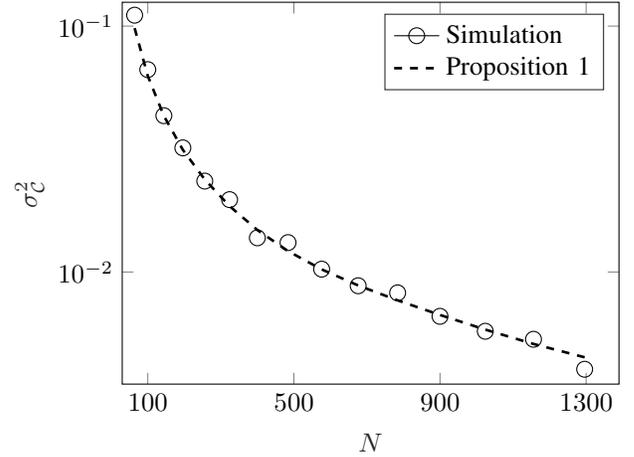

\begin{proposition}
	\label{th:0-1}
	Let the sub-linearly growing maximum eigenvalue of $\mR$ be uniformly bounded from above as
	\begin{align} 
	\lambda_{\max} \leq a N^u,
	\end{align}
	for some non-negative real $a$ and $0 \leq u < 1$. Then, there exists an integer $N_0$, such that for $N \geq N_0$ the mean and variance of the maximum mutual information is bounded uniformly as 
	\begin{subequations}
		\begin{align}
			\mu_{\mathcal{C}} &\geq b + \brc{1-q}  \log_2 N\\
			\sigma_{\mathcal{C}}^2 &\leq {\frac{c}{N^{1-u}}},
		\end{align}
	\end{subequations}
for some non-negative real scalars $b$ and $c$.
\end{proposition}
\begin{proof}
	We start the proof by noting that 
	\begin{subequations}
		\begin{align}
			\bar{\alpha}_N &=  \frac{\alpha_{\mathrm{ r}} \alpha_{\mathrm{ s}} }{1+\kappa_\mathrm{r}} A_N^2\\
			&= \frac{\alpha_{\mathrm{ r}} \alpha_{\mathrm{ s}} }{1+\kappa_\mathrm{r}} \frac{A_{\rm IRS}^2}{N^2}
			= \frac{\alpha_{\mathrm{ r}} \alpha_{\mathrm{ s}} }{1+\kappa_\mathrm{r}} \frac{A_0^2}{N^{2q}}. \label{eq:aplpha_N}
		\end{align}
	\end{subequations}
	From Proposition~\ref{th:0}, since $\mu \geq \bar{\alpha}_N\kappa_{\mathrm{ r} }N^2$, we have 
\begin{align}
	\mu_{\mathcal{C}} 
	&\geq   \log_2 \left( \rho M \kappa_{\mathrm{ r} } \bar{\alpha}_N N^2 \right).
\end{align}
Using  \eqref{eq:aplpha_N}, we can hence conclude that
\begin{align}
	\mu_{\mathcal{C}} 
	&= \log_2 \left( \frac{\alpha_{\mathrm{ r}} \alpha_{\mathrm{ s}} }{1+\kappa_\mathrm{r}} \rho M  \kappa_{\mathrm{ r} } A_0^2  \right) + \log_2 N^{2-2q}.
\end{align}
Noting that the first term on the right hand side does not grow in $N$, we conclude that 
\begin{align}
	\mu_{\mathcal{C}} &\geq  b + \brc{1-q}\log_2 N,
\end{align}
by setting $b$ to 
\begin{align}
	b = \log_2 \left( \frac{\alpha_{\mathrm{ r}} \alpha_{\mathrm{ s}} }{1+\kappa_\mathrm{r}} \rho M  \kappa_{\mathrm{ r} } A_0^2  \right) .
\end{align}

To further bound the variance, we note that $\norm{\mathbf{\bar{h} }_{\mathrm{ r}}}^2 = N$. As a result, $\mathbf{\bar{h}}_{\mathrm{r}}^\her \mR \mathbf{\bar{h}}_{\mathrm{r}}/N$ determines the \textit{Rayleigh quotient} of $\mR$ at $\mathbf{\bar{h} }_{\mathrm{ r}}$ which is bounded from above by the maximum eigenvalue of $\mR$ \cite{horn2013matrix}. We hence have
\begin{align}
\bar{\alpha}_N	\mathbf{\bar{h}}_{\mathrm{r}}^\her \mR \mathbf{\bar{h}}_{\mathrm{r}}&\leq   \bar{\alpha}_N\lambda_{\max} N.
\end{align}
Let $\lambda_{\max}$ be bounded uniformly from above by $aN^u$ for some $0\leq u <1$. This concludes that $\mu$, $\eta$ and $\omega$ in Proposition~\ref{th:0} are uniformly bounded from above respectively by 
\begin{subequations}
	\begin{align}
		\mu^{\rm U } \brc{N} &= \mathcal{O}\brc{ N^{2-2q} }\\
		\eta^{\rm U } \brc{N} &= \mathcal{O}\brc{ N^{1+u-2q} }\\
		\omega^{\rm U } \brc{N} &= \mathcal{O}\brc{ N^{2-2q} }.
	\end{align}
\end{subequations}

We now start with bounding $\sigma_{\mathcal{C}}$. Since $1+ \rho M \mu >  \rho M \mu$, we can initially bound the variance as
	\begin{align}
		\sigma_{\mathcal{C}} &< \frac{ \displaystyle  \log_2 \mathrm{e} }{\mu }  \sqrt{ \omega \eta + \eta + \frac{M-1}{M} \alpha_{\mathrm{ d}} A_M } =\Xi_N \sqrt{\frac{ \eta  }{\mu }} 
	\end{align}
where $\Xi_N$ is given by
	\begin{align}
	\Xi_N=   \sqrt{ \frac{1}{\mu \eta}\brc{ \omega \eta + \eta + \frac{M-1}{M} \alpha_{\mathrm{ d}} A_M } } \log_2 \mathrm{e}
\end{align}
The term $\Xi_N$ is uniformly bounded from above by a constant. We can hence conclude that there exists an integer $N_1$ and a non-negative real $c_0$, such that for $N\geq N_1$, we have
\begin{align}
	\sigma_{\mathcal{C}} 
	&\leq c_0  \sqrt{ \frac{ \eta }{ \mu  } } 
\end{align}
Noting that $\mu \geq \kappa_{\mathrm{ r} } \bar{\alpha}_N N^2$, we can further bound the above upper-bound as
\begin{align}
	\sigma_{\mathcal{C}} 
	&\leq c_1 \frac{\sqrt{\eta}}{N^{1-q}} \label{upper:1}
\end{align}
where $c_1$ being defined as 
\begin{align}
	c_1 = \sqrt{\frac{c_0 \brc{1+\kappa_{\mathrm{ r} }}}{\alpha_{\mathrm{ r}} \alpha_{\mathrm{ s}} \kappa_{\mathrm{ r} }}}.
\end{align}


Since $\eta$ is uniformly bounded with $\eta^{\rm U } \brc{N}$, we conclude that there exists an integer $N_2 \geq N_1$ and a non-negative real $c_2$, such that for $N \geq N_2$, we have 
		\begin{align}
\frac{ \eta }{  N^{2-2q}  } \leq \frac{c_2 }{N^{1-u}}.
\end{align}
The proof is finally completed by setting $c = c_1^2 c_2$ and $N_0$ to be $N_0 = N_2$.
\end{proof}

Proposition~\ref{th:0-1} formulates an intuitive behavior. The higher the correlation among the \ac{irs} elements is, i.e., larger $u$, the slower the end-to-end channel hardens. It further indicates~that the fastest hardening is achieved by $u=q=0$.~This~corresponds to reflecting arrays whose number of \textit{strongly} correlated elements does not scale with $N$ while the area of the \ac{irs} grows linearly in $N$. From implementation viewpoint, this means that the neighboring reflecting elements on the \ac{irs} are well-distanced, and the physical dimensions of the \ac{irs} grow with the number of reflecting elements, such that the distance between two neighboring elements remains constant.

We now validate Proposition~\ref{th:0-1} through a numerical experiment. We consider the covariance matrix derived in \cite[Proposition~1]{bjornson2020rayleigh} for a planar \ac{irs} whose elements are spaced with $\lambda/2$ in both directions, i.e., $d_\xx = d_\yy = \lambda/2$. This means that for this array $q=0$. The covariance matrix in this case is specified via \eqref{eq:RR}. We consider a square array, i.e., $N_\xx = N_\yy$ and let $N$ grow gradually from $N= 64$ to $N=1296$. For this sequence of covariance matrices, $\lambda_{\max}$ is plotted against $N$ in Fig.~\ref{fig:lam_max}. We further use the curve fitting toolbox of MATLAB, i.e., \texttt{cftool} \cite{matlabcurvefitting}, to fit the collected data to a curve of form $\lambda_{\max} = a N^u$. The result is shown with a dashed line in the figure for which $a=0.83$ and $u=0.25$. 

 \begin{figure}
	\centering
	\input{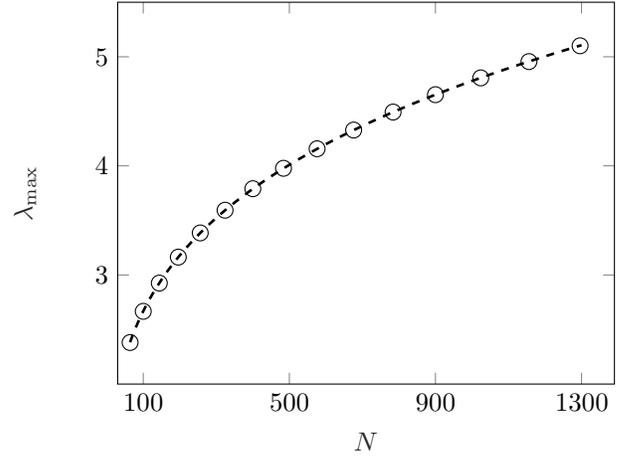}
	\caption{$\lambda_{\max}$ against $N$ for the classic \ac{irs} covariance matrix derived in \cite[Proposition~1]{bjornson2020rayleigh}. The dashed line shows the fitted curve which reads $\lambda_{\max} \approx 0.83\sqrt[4]{N}$, i.e., $u = 0.25$.}
	\label{fig:lam_max}
\end{figure}

Given the results in Fig.~\ref{fig:lam_max}, Proposition~\ref{th:0-1} suggests that~the variance of the mutual information in this case is uniformly bounded from above by $cN^{-0.75}$ for some $c$. This is shown in Fig.~\ref{fig:upp}, where we compare $\sigma_{\mathcal{C}}^2$ against the upper bound with $c=1.9$. The numerical results show consistency~with~Proposi-tion~\ref{th:0-1}. Interestingly, the suggested upper bound gives a \textit{pessimistic} approximation of the hardening speed. In fact the true variance, drops much faster than the order of the upper bound.

\begin{figure}
	\centering
	\input{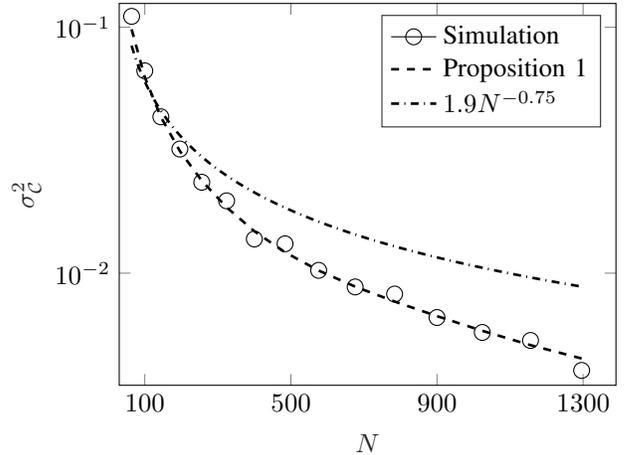}
	\caption{Comparing the channel hardening speed to the uniform upper bound suggested by Proposition~\ref{th:0-1}.}
	\label{fig:upp}
\end{figure}

\subsection{Scaling of IRS Area and Covariance Matrix}
As indicated earlier, $u$ and $q$ are mutually coupled. This is easily seen by sending $q\rightarrow 1$. In this case, the total area of the \ac{irs} is constant. This means that in the limit $N\rightarrow \infty$, the covariance matrix converges to a rank-one matrix and hence
\begin{align}
	\lambda_{\max} = \tr{\mR} = N,
\end{align}
or equivalently $u\rightarrow 1$; see also \cite{bjornson2020rayleigh}. To understand this mutual coupling analytically and find out how realistic the condition \eqref{eq:const_IRS} is, i.e., $u\geq q$, we conduct a numerical experiment. For a given $q$, we repeat the previous experiment. This means that we consider square \acp{irs} of size $N$, i.e., $N_\xx = N_\yy = \sqrt{N}$, and let $N$ grow gradually from $N= 64$ to $N=1296$. The distance between two neighboring elements on the \ac{irs} in this case is
\begin{align}
	d_\xx = d_\yy = \frac{\lambda}{2N^{q/2}}.
\end{align}
Invoking the least-squares method, we numerically find $u$ such that $\lambda_{\max} \leq aN^u$ for some real $a$. We then change $q$ from $q=0$ to $q=1$ and plot $u$ against it. The result is shown in Fig.~\ref{fig:u_vs_q} along with the region $u \geq q$ which corresponds to the constraint given by \eqref{eq:const_IRS}.

\begin{figure}
	\centering
	\begin{tikzpicture}

\begin{axis}[%
	width=2.6in,
	height=2in,
	at={(1.262in,0.7in)},
	scale only axis,
	xmin=-.05,
	xmax=1.05,
	xtick={0,1},
	xticklabels={{$0$},{$1$}},
	xlabel style={font=\color{white!15!black}},
	xlabel={$q$},
	ymin=-.05,
	ymax=1.05,
	ytick={0,1},
	yticklabels={{$0$},{$1$}},
	ylabel style={font=\color{white!15!black}},
	ylabel={$u$},
	yminorticks=true,
	axis background/.style={fill=white},
	legend style={at={(.97,.97)},legend cell align=left, align=left, draw=white!15!black}
	]
\addplot [color=black, dashed, line width=1.0pt,forget plot]
  table[row sep=crcr]{%
0	0.252938221766722\\
0.025	0.266828009551138\\
0.05	0.288776230761314\\
0.075	0.308271534046978\\
0.1	0.326811208525004\\
0.125	0.34640395603301\\
0.15	0.365155546972958\\
0.175	0.383706141477129\\
0.2	0.402838996931415\\
0.225	0.421395915991257\\
0.25	0.439953971966012\\
0.275	0.458531373971296\\
0.3	0.477110045586087\\
0.325	0.495707289196655\\
0.35	0.5143221349567\\
0.375	0.532948573285285\\
0.4	0.551599037312229\\
0.425	0.570317519873392\\
0.45	0.589098666647782\\
0.475	0.608019092234421\\
0.5	0.626961442153123\\
0.525	0.645645171113727\\
0.55	0.664956058728108\\
0.575	0.683849607533629\\
0.6	0.701991120557559\\
0.625	0.721648453367288\\
0.65	0.741653407004811\\
0.675	0.760746352409074\\
0.7	0.77967488718215\\
0.725	0.799742050350287\\
0.75	0.820302135872291\\
0.775	0.838524174089761\\
0.8	0.858452706415122\\
0.825	0.882347869688267\\
0.85	0.905370840950944\\
0.875	0.927049207458561\\
0.9	0.947068413024721\\
0.925	0.965248020705397\\
0.95	0.981523231926511\\
0.975	0.995921565659442\\
1	0.99999998922787\\
};

\addplot [color=black, dotted , line width=1.0pt,forget plot]
  table[row sep=crcr]{%
 -.5 -.5\\
0	0\\
0.025	0.025\\
0.05	0.05\\
0.075	0.075\\
0.1	0.1\\
0.125	0.125\\
0.15	0.15\\
0.175	0.175\\
0.2	0.2\\
0.225	0.225\\
0.25	0.25\\
0.275	0.275\\
0.3	0.3\\
0.325	0.325\\
0.35	0.35\\
0.375	0.375\\
0.4	0.4\\
0.425	0.425\\
0.45	0.45\\
0.475	0.475\\
0.5	0.5\\
0.525	0.525\\
0.55	0.55\\
0.575	0.575\\
0.6	0.6\\
0.625	0.625\\
0.65	0.65\\
0.675	0.675\\
0.7	0.7\\
0.725	0.725\\
0.75	0.75\\
0.775	0.775\\
0.8	0.8\\
0.825	0.825\\
0.85	0.85\\
0.875	0.875\\
0.9	0.9\\
0.925	0.925\\
0.95	0.95\\
0.975	0.975\\
1	1\\
1.05 1.05\\
};

\addplot [draw=none,fill=gray, fill opacity=0.25]
coordinates {
	(1.5, 1.5) 
	(-.5, 1.5)
	(-.5, -.5)  };

\end{axis}
\end{tikzpicture}%
	\caption{Eigenvalue order exponent $u$ against the area order exponent $q$. The gray area is the feasible region $u\geq q$ where \eqref{eq:const_IRS} is satisfied.}
	\label{fig:u_vs_q}
\end{figure}
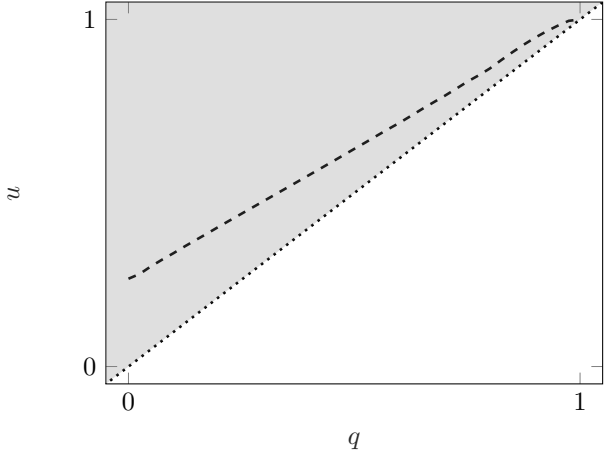

As the figure suggests, for the covariance matrix, derived for the Rayleigh fading model, $u\geq q$. This is intuitive, since in both extreme cases of $q =0$ and $q\rightarrow1$, this constraint is satisfied. The result further show a linear growth of $u$ in $q$. This observation further suggests that Propositions~\ref{th:0} and \ref{th:0-1} are further valid approximations in the limits $q\rightarrow 1$ and $u\rightarrow 1$. We confirm this intuition through numerical analysis of this limiting case in the sequel.

\subsection{The Extreme Case of Rank-One Covariance}
\label{sec:Linear}
As mentioned, the sub-linearity constraint in Propositions~\ref{th:0} and \ref{th:0-1} is not a strong constraint. In fact, noting that $\lambda_{\max} \leq N$, one can conclude that the sub-linearity constraint is satisfied in all settings, but the \textit{extreme} case with a \textit{finite-rank} covariance matrix whose rank does not grow with $N$, and therefore $\lambda_{\max}$ scales linearly in $N$. From the mutual coupling between $u$ and $q$, in this case we further have $q=1$ which corresponds to an \ac{irs} whose total area is kept fixed. 

Inspired by our observation in Section~\ref{sec:IRS_Size}, we conjecture that Propositions~\ref{th:0} and \ref{th:0-1} give accurate characterizations, even in an extreme case with a linearly growing $\lambda_{\max}$. To confirm this conjecture, we run simulations for the same setting considered in Section~\ref{sec:Numerical} while replacing the covariance matrix with $\mR = \boldsymbol{1}_N$, i.e., the matrix of all ones, and correspondingly setting $q=1$, i.e., assuming
\begin{align}
	A_N = \frac{A_0}{N}.
\end{align}
for some real $A_0$. To achieve this scaling, we set the distance between two neighboring \ac{irs} elements to be
\begin{align}
	d_\xx = d_\yy = \frac{\lambda}{2\sqrt{N}}
\end{align}
which means that $A_0 = \lambda^2/4$. We further set $\alpha_{\mathrm{ r}} = \alpha_{\mathrm{ s}} = 1/A_0$ to make the comparison to the first experiment in Section~\ref{sec:Numerical} fair. The remaining parameters are the same as those considered in Section~\ref{sec:Numerical}.

The matrix $\mR = \boldsymbol{1}_N$ is a rank one matrix whose maximum eigenvalue reads $\lambda_{\max} = N$, and thus scales linearly with $N$. One can observe this case, as an extreme case of the covariance matrix in \eqref{eq:RR}, in which the distance between the most outer two elements of the \ac{irs} is still significantly smaller than the wavelength\footnote{In other words, an \ac{irs} whose physical dimensions are significantly smaller than the wavelength.}.

Fig.~\ref{fig:hist2} shows the numerically-evaluated histogram of $\mathcal{C}$, as well as the properly scaled version of the fitted Gaussian probability density and the distribution suggested by Proposition~\ref{th:0}. The histogram is evaluated for $10^5$ channel realizations. As the figure shows, the histogram in this case is rather loosely approximated by the distribution given in Proposition~\ref{th:0}, and compared to the almost-perfect match in Figs.~\ref{fig:hist} and \ref{fig:dist}, it shows some approximation error. This approximation error follows from the lower speed of channel hardening in this case. In fact, from Proposition~\ref{th:0-1}, we know that by sending $q,u\rightarrow 1$, the mean and variance of $\mathcal{C}$ scale with $\mathcal{O}\brc{1}$ in $N$, and hence, the approximation in Proposition~\ref{th:0} becomes inaccurate\footnote{This is the main reason that Proposition~\ref{th:0} excludes the case $q=u=1$.}.

\begin{figure}
	\centering
	\input{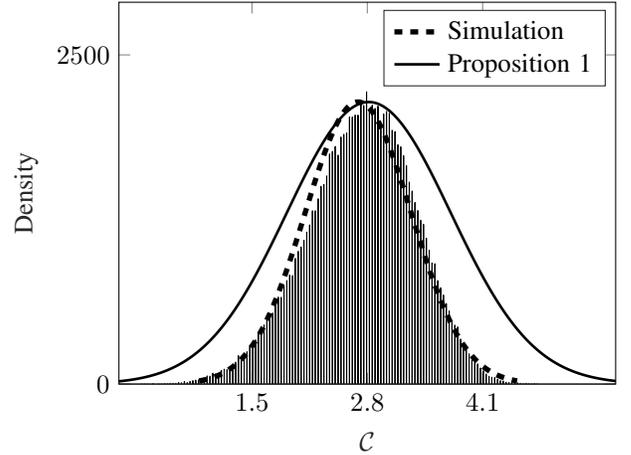}
	\caption{Histogram and the fitted Gaussian density of $\mathcal{C}$ for the extreme case of $\mR = \boldsymbol{1}_N$.  The scaled version of the analytical density function given by Proposition~\ref{th:0} is further plotted with a solid line.}
	\label{fig:hist2}
\end{figure}

Although the approximation by Proposition~\ref{th:0} is inaccurate, the histogram in Fig.~\ref{fig:hist2} shows a rather small standard deviation for $\mathcal{C}$. This is an interesting observation confirming our earlier finding in Section~\ref{sec:IRS_Size} which indicates that the upper bound given by Proposition~\ref{th:0-1} is rather \textit{pessimistic}.This finding can be analytically supported as follows: The sub-linearity constraint in Propositions~\ref{th:0} and \ref{th:0-1} comes from utilizing the inequality  $	\mathbf{\bar{h}}_{\mathrm{r}}^\her \mR \mathbf{\bar{h}}_{\mathrm{r}} \leq   \lambda_{\max} N$ to bound the variance $\sigma_{\mathcal{C}}^2$. Nevertheless, for rank-deficient covariance matrices, this bound is often asymptotically very loose. This means that despite the inaccuracy of Proposition~\ref{th:0} for exactly linear scaling, i.e., $q=u=1$, the channel hardening still occurs in this case. 

We now invoke the above heuristic conclusion to further modify Proposition~\ref{th:0} for linearly scaling $A_{\rm IRS}$ and $\lambda_{\max}$. In fact, assuming that $\mathcal{C}$ is still Gaussian in this case, we can conclude that the approximation error in this case comes from the inaccuracy of the variance in  Proposition~\ref{th:0} for cases with linear scaling. The following pseudo-proposition gives a modified approximation for the distribution of $\mathcal{C}$ which is also accurate for linearly scaling scenarios. Interestingly, by assuming sub-linear scaling, the approximation reduces to the limiting distribution given by Proposition~\ref{th:0}.

\begin{pproposition}
	\label{th:0-2}
	For linearly scaling $A_{\rm IRS}$ and $\lambda_{\max}$, i.e., cases with $q=u=1$, the maximum mutual information $\mathcal{C}$ is well approximated by a real Gaussian random variable whose mean is $\mu_{\mathcal{C}}$ and whose variance is
		\begin{align}
			\hat{\sigma}^2_{\mathcal{C}} &=  \frac{\kappa_{\mathrm{ r} }}{\kappa_{\mathrm{ r} } + \vartheta }  \sigma^2_{\mathcal{C}} ,
		\end{align}
where $\mu_{\mathcal{C}}$ and $\sigma^2_{\mathcal{C}}$ are given in Proposition~\ref{th:0}, and 
\begin{align}
	\vartheta = \frac{\alpha_{\mathrm{ d}} A_M}{2 \bar{\alpha}_N N^2} \label{eq:vartheta}
\end{align}
\end{pproposition}
\begin{proof}
	The proof is given in Section~\ref{sec:Proof_PPro}.
\end{proof}

For the numerical experiment of Fig.~\ref{fig:hist2}, we further plot the approximation by Pseudo-Proposition~\ref{th:0-2} in Fig.~\ref{fig:hist3}. As the figure demonstrates, the proposed approximated density tracks the empirically-evaluated histogram more closely\footnote{The interested reader can check that for $u,q<1$, we have $\hat{\sigma}_{\mathcal{C}}^2 \rightarrow \sigma_{\mathcal{C}}^2$, as we set $N\rightarrow \infty$.}~compared~with the limiting result given by Proposition~\ref{th:0}.

\begin{figure}
	\centering
	\input{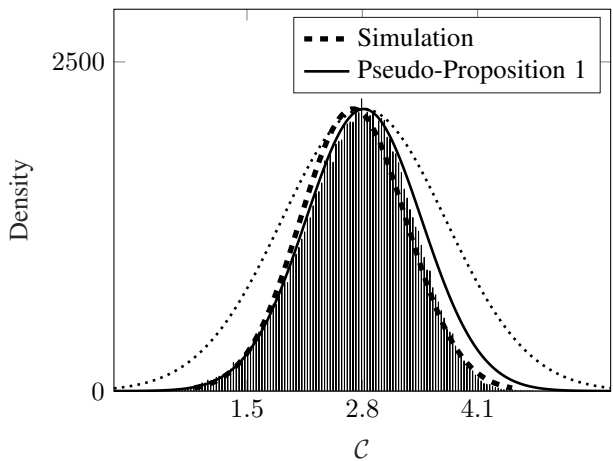}
	\caption{Histogram and the fitted Gaussian density of $\mathcal{C}$ for the extreme case of $\mR = \boldsymbol{1}_N$.  The scaled version of the density functions given by Pseudo-Proposition~\ref{th:0-2} and Proposition~\ref{th:0} are plotted with solid and dotted lines, respectively.}
	\label{fig:hist3}
\end{figure}


\section{IRS-Transmitter Dimensional Trade-Off}
\label{sec:trade}
The analytical results of this study enable us to~address~various design challenges in \ac{irs}-aided systems. In this section, we focus on a particular application: We employ the results to investigate the dimensional trade-off between the \ac{bs} and the \ac{irs}. More precisely, we try to answer this fundamental question: How does the transmit array dimension, i.e., $M$, change, when we employ an \ac{irs} to enhance the communication link?

\subsection{Dimensional Trade-Off for Ergodic Capacity}
We start our investigations by considering cases in which the \ac{nlos} links experience a fast fading process. This means that the channel varies from a short block of symbol transmit intervals to another. In this case, the performance~is~best~described by the \textit{ergodic capacity}, defined as
\begin{align}
	\bar{\mathcal{C}} = \Ex{\mathcal{C}}.
\end{align}

For a target ergodic capacity, the \ac{bs} needs to be equipped with a certain number of transmit antennas, i.e., $M$. Intuitively, this number of required antennas is expected to reduce, as we increase the \ac{irs} dimension. This draw a dimensional trade-off between $N$ and $M$: With a larger \ac{irs}, the minimum number of required \ac{bs} antennas to achieve a given target performance decreases. Proposition~\ref{th:0} enables us to quantitatively formulate this dimensional trade-off. To this end, we first use Proposition~\ref{th:0} and derive the ergodic capacity for a given $N$ and $M$ in a closed form as follows:
\begin{align}
\hspace*{-1mm}	\bar{\mathcal{C}} =\log_2\brc{1+\rho M \dbc{\alpha_{\mathrm{ d}} A_M + \kappa_{\mathrm{ r} } \bar{\alpha}_N N^2 + \bar{\alpha}_N \mathbf{\bar{h}}_{\mathrm{r}}^\her \mR \mathbf{\bar{h}}_{\mathrm{r}}}}. \hspace*{-1mm}
\end{align}
It is hence readily concluded that for a given $N$ and the target ergodic capacity $\bar{\mathcal{C}}$, $M$ needs to satisfy
\begin{align}
	M \geq \frac{2^{\bar{\mathcal{C}} } - 1}{
	\rho \brc{\alpha_{\mathrm{ d}} A_M + \kappa_{\mathrm{ r} } \bar{\alpha}_N N^2 + \bar{\alpha}_N \mathbf{\bar{h}}_{\mathrm{r}}^\her \mR \mathbf{\bar{h}}_{\mathrm{r}}}
}. \label{eq:Min_M_Erg}
\end{align}
The inequality in \eqref{eq:Min_M_Erg} formulates the trade-off between $N$ and $M$: By increasing the physical dimensions of the \ac{irs}, a smaller array is required at the \ac{bs}. It further specifies the speed of this drop, i.e., $M$ drops proportional to $N^{2q-2}$. 

We now conduct some experiments to investigate this trade-off numerically. To this end, a setting is considered in which an \ac{irs} is distanced from the \ac{bs} with $D_{\rm s} = 25$ m.  The receiver is further located at a random point at which its distances from the \ac{bs} and the \ac{irs} are $D_{\rm d} = 20$ m and $D_{\rm d} = 15$ m, respectively. We further set the \ac{aoa} and \acp{aod} to be the same as those considered in Section~\ref{sec:Numerical}, and assume that $\rho=1$. The large-scale fading coefficients are generated as
\begin{align}
	\alpha_i = \frac{\alpha_{\rm ref}}{D_i^{\varepsilon_i}}
\end{align}
for $i\in\set{\mathrm{ s},\mathrm{ d},\mathrm{ r}}$, where $\varepsilon_i$ is the path-loss exponent of link $i$ and $\alpha_{\rm ref}$ is the reference path-loss. For numerical simulations, we set $\log \alpha_{\rm ref} = 10$ dB, $\varepsilon_{\mathrm{ d}} = 3.5$ and $\varepsilon_{\mathrm{s}} = \varepsilon_{\mathrm{r}} = 2.3$.

The \ac{irs} is considered to be a square planar array, i.e., $N_\xx = N_\yy$, whose elements are distanced with a half wavelength. This means that $q=0$. The covariance matrix of the \ac{nlos} link is further set to be the Rayleigh covariance matrix, i.e., \eqref{eq:RR}. We assume that the phase-shifts of the reflecting elements are set according to Proposition~\ref{th:0}; hence, the configuration of the transmit array is not required to be taken into account.

To investigate the dimensional trade-off between $M$ and $N$, we let the ergodic capacity to be set to a given target value $\bar{\mathcal{C}}$. The number of reflecting elements is then increased gradually starting from $N=64$ and ending at $N=1296$. For each $N$, we determine the \textit{real-valued} lower bound on the minimum required $M$, denoted by $M_{\rm Erg}$, from \eqref{eq:Min_M_Erg}, i.e., 
\begin{align}
	M_{\rm Erg} = \frac{2^{\bar{\mathcal{C}} } - 1}{
		\rho \brc{\alpha_{\mathrm{ d}} A_M + \kappa_{\mathrm{ r} } \bar{\alpha}_N N^2 + \bar{\alpha}_N \mathbf{\bar{h}}_{\mathrm{r}}^\her \mR \mathbf{\bar{h}}_{\mathrm{r}} }
	}.
\end{align}

Fig.~\ref{fig:erg} plots $M_{\rm Erg}$ against the \ac{irs} size $N$, for three various choices of $\bar{\mathcal{C}}$. The minimum required \ac{bs} array size, denoted by $M_{\min}$, is then determined by quantizing $M_{\rm Erg}$ to the next larger integer.  As the figure shows, for rather large choices of $N$, the target ergodic capacity with only one antenna at the transmitter. Our further numerical investigations confirm the consistency of the result with simulations.

\begin{figure}
	\centering
%
%
\begin{tikzpicture}
	
	\begin{axis}[%
		width=2.6in,
		height=2in,
		at={(1.262in,0.697in)},
		scale only axis,
		bar shift auto,
		xmin=50,
		xmax=1350,
		xtick={100,500,900,1300},
		xticklabels={{$100$},{$500$},{$900$},{$1300$}},
		xlabel style={font=\color{white!15!black}},
		xlabel={$N$},
		ymode=log,
		ymin=0.005,
		ymax=200,
		ylabel style={font=\color{white!15!black}},
		ylabel={$M_{\rm Erg}$},
		axis background/.style={fill=white},
		legend style={legend cell align=left, align=left, draw=white!15!black}
		]

\addplot [color=black, dashdotted, line width=1.0pt]
  table[row sep=crcr]{%
64	3.96294667164879\\
81	2.48769956738759\\
100	1.63861264028754\\
121	1.12242244751533\\
144	0.794220159828345\\
169	0.577583419242135\\
196	0.429974709082893\\
225	0.326621443921504\\
256	0.252522778232451\\
289	0.198285574783479\\
324	0.157853006304537\\
361	0.127216571516936\\
400	0.103661981631294\\
441	0.08531303765341\\
484	0.0708485661578969\\
529	0.0593225667319003\\
576	0.0500471980546759\\
625	0.0425153139937137\\
676	0.0363481990762443\\
729	0.0312596862204255\\
784	0.0270311980942375\\
841	0.023494103402754\\
900	0.020517089418333\\
961	0.0179970250537501\\
1024	0.0158522592532851\\
1089	0.0140176626034524\\
1156	0.0124409203913309\\
1225	0.0110797313740419\\
1296	0.00989967782054624\\
};
\addlegendentry{$\bar{\mathcal{C}} = 1$}

\addplot [color=black, dashed, line width=1.0pt]
  table[row sep=crcr]{%
64	11.8888400149464\\
81	7.46309870216278\\
100	4.91583792086263\\
121	3.36726734254599\\
144	2.38266047948503\\
169	1.73275025772641\\
196	1.28992412724868\\
225	0.979864331764512\\
256	0.757568334697354\\
289	0.594856724350436\\
324	0.473559018913612\\
361	0.381649714550807\\
400	0.310985944893882\\
441	0.25593911296023\\
484	0.212545698473691\\
529	0.177967700195701\\
576	0.150141594164028\\
625	0.127545941981141\\
676	0.109044597228733\\
729	0.0937790586612764\\
784	0.0810935942827125\\
841	0.070482310208262\\
900	0.061551268254999\\
961	0.0539910751612502\\
1024	0.0475567777598552\\
1089	0.0420529878103572\\
1156	0.0373227611739927\\
1225	0.0332391941221256\\
1296	0.0296990334616387\\
};
\addlegendentry{$\bar{\mathcal{C}} = 2$}

\addplot [color=black, line width=1.0pt]
  table[row sep=crcr]{%
64	27.7406267015415\\
81	17.4138969717131\\
100	11.4702884820128\\
121	7.85695713260731\\
144	5.55954111879841\\
169	4.04308393469495\\
196	3.00982296358025\\
225	2.28635010745053\\
256	1.76765944762716\\
289	1.38799902348435\\
324	1.10497104413176\\
361	0.890516000618551\\
400	0.725633871419058\\
441	0.59719126357387\\
484	0.495939963105278\\
529	0.415257967123302\\
576	0.350330386382732\\
625	0.297607197955996\\
676	0.25443739353371\\
729	0.218817803542978\\
784	0.189218386659662\\
841	0.164458723819278\\
900	0.143619625928331\\
961	0.12597917537625\\
1024	0.110965814772996\\
1089	0.0981236382241668\\
1156	0.0870864427393163\\
1225	0.0775581196182932\\
1296	0.0692977447438236\\
};
\addlegendentry{$\bar{\mathcal{C}} = 3$}

\addplot [color=black, dotted, mark size=3.0pt, mark=o, mark options={solid, black}]
  table[row sep=crcr]{%
64	28\\
81	18\\
100	12\\
121	8\\
144	6\\
169	5\\
196	4\\
225	3\\
256	2\\
289	2\\
324	2\\
361	1\\
400	1\\
441	1\\
484	1\\
529	1\\
576	1\\
625	1\\
676	1\\
729	1\\
784	1\\
841	1\\
900	1\\
961	1\\
1024	1\\
1089	1\\
1156	1\\
1225	1\\
1296	1\\
};
\addlegendentry{$M_{\min}$ for $\bar{\mathcal{C}} = 3$}

\end{axis}
\end{tikzpicture}%
	\caption{IRS-transmitter dimensional trade-off for target ergodic capacity $\bar{\mathcal{C}}$: The minimum required number of transmit antennas $M_{\min}$ is determined by quantizing $M_{\rm Erg}$ to the next larger integer.}
	\label{fig:erg}
\end{figure}
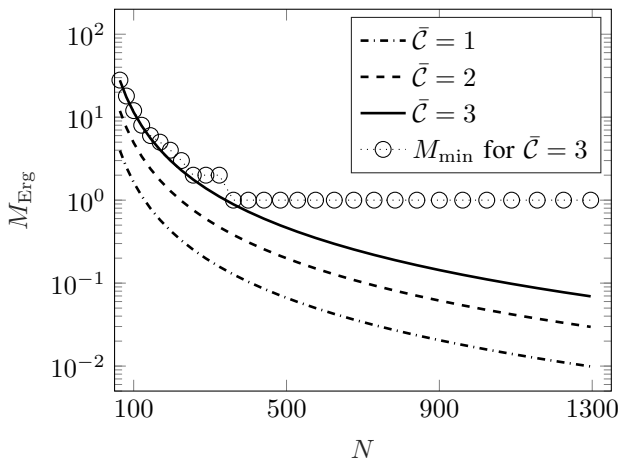

\subsection{Dimensional Trade-Off for Outage Probability}
With a slow fading process, the channel remains fixed within a long sequence of symbol intervals. The appropriate metric for performance evaluation in this case is therefore the outage probability which is defined for a given target rate $\mathcal{R}$ as \cite{bolcskei2006mimo} 
\begin{align}
	\mathcal{P}_\mathrm{out} \brc{ \mathcal{R}} =\Pr \left\lbrace 	\mathcal{C} < \mathcal{R}\right\rbrace .
\end{align}

The dimensional trade-off can be also studied for this metric using the analytical results of Section~\ref{sec:main}. From Proposition~\ref{th:0}, we can determine the outage probability as 
\begin{align}
{\mathcal{P}}_\mathrm{out} \brc{ \mathcal{R}}= \mathrm{Q}\left( \frac{\mu_{\mathcal{C}} -\mathcal{R}}{\sigma_{\mathcal{C}} }\right),
\end{align}
with $\mathrm{Q}\brc{\cdot}$ being the standard $\mathrm{Q}$-function. To guarantee achieving ${\mathcal{P}}_\mathrm{out} \brc{ \mathcal{R}} \leq p_{\rm out}$ for a given target rate $\mathcal{R}$, we need to have
\begin{align}
\frac{\mu_{\mathcal{C}} -\mathcal{R}}{\sigma_{\mathcal{C}} } \geq  \mathrm{Q}^{-1} \brc{p_{\rm out}},
\end{align}
where $\mathrm{Q}^{-1}\brc{\cdot}$ is the inverse of the $\mathrm{Q}$-function with respect to composition, i.e., $\mathrm{Q}^{-1}\brc{\mathrm{Q}^{-1} \brc{p_{\rm out}}} = p_{\rm out}$.

For a given reflecting array and rate $\mathcal{R}$, it is readily shown, after few lines of calculations, that the target outage is achieved if $M\geq M_{\rm Out}$, where $M_{\rm Out}$ a positive real-valued solution to the following fixed-point equation in $\xx$:
\begin{align}
	\mathcal{R} + \frac{C}{1 + \rho \mu \xx} \sqrt{B \xx^2 + \omega\alpha_{\mathrm{ d}} A_M \xx} = \log_2\brc{ 1+ \rho \mu \xx}.
\end{align}
Here, $\mu$ and $\omega$ are defined for the given \ac{irs} as in Proposition~\ref{th:0}. Moreover, $C$ and $B$ are given by
\begin{subequations}
	\begin{align}
	C &= \rho\mathrm{Q}^{-1} \brc{p_{\rm out}} \log_2 \mathrm{e},\\
	B &= \brc{\omega+1} \mathbf{\bar{h}}_{\mathrm{r}}^\her \mR \mathbf{\bar{h}}_{\mathrm{r}} + \alpha_{\mathrm{ d}} A_M.
\end{align}
\end{subequations}

Fig.~\ref{fig:outage} demonstrates the \ac{irs}-\ac{bs} dimensional trade-off with respect to the outage probability for various target outages. Here, the setting is set to be the same as the one considered in Fig.~\ref{fig:erg}. The target rate is moreover set to $\mathcal{R} = 3$. For $50\%$ outage, the trade-off curve recovers the one given in Fig.~\ref{fig:erg} for $\bar{\mathcal{C}} = 3$. This follows from the symmetry of the Gaussian density which leads to this property that the median coincides with the mean. We further observe that by reducing the outage probability, the trade-off figure shifts upward. This is intuitive, as for lower-outages, we require a better end-to-end link. These upward shifts are however not significant and rather negligible for large choices of $N$. This is a direct result of the channel hardening: With large $N$, the capacity term $\mathcal{C}$ is almost deterministic, and hence the outage probability tends to a step function in $\mathcal{R}$. 

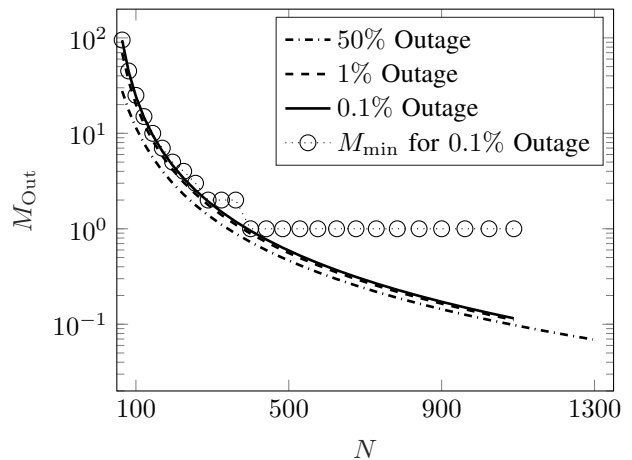
\begin{figure}
	\centering
%
%
\begin{tikzpicture}
	
		\begin{axis}[%
		width=2.6in,
		height=2in,
		at={(1.262in,0.697in)},
		scale only axis,
		bar shift auto,
		xmin=50,
		xmax=1350,
		xtick={100,500,900,1300},
		xticklabels={{$100$},{$500$},{$900$},{$1300$}},
		xlabel style={font=\color{white!15!black}},
		xlabel={$N$},
		ymode=log,
		ymin=0.02,
		ymax=200,
		ylabel style={font=\color{white!15!black}},
		ylabel={$M_{\rm Out}$},
		axis background/.style={fill=white},
		legend style={legend cell align=left, align=left, draw=white!15!black}
		]

\addplot [color=black, dashdotted, line width=1.0pt]
  table[row sep=crcr]{%
64	27.7406263268219\\
81	17.4138382379046\\
100	11.4702882662438\\
121	7.85695710882592\\
144	5.55954016222303\\
169	4.04308359083247\\
196	3.00982059211461\\
225	2.2863501015746\\
256	1.76765944758472\\
289	1.38799879157353\\
324	1.10497104410665\\
361	0.890516000494232\\
400	0.72563387141881\\
441	0.59719126209925\\
484	0.495939963104891\\
529	0.415257887248035\\
576	0.350330386372158\\
625	0.297607196609372\\
676	0.254437393430701\\
729	0.218817803538831\\
784	0.189218386659606\\
841	0.164458716206014\\
900	0.143619625887153\\
961	0.125979175337749\\
1024	0.110965814753719\\
1089	0.0981236382241665\\
1156	0.0870864427393162\\
1225	0.0775581196182932\\
1296	0.0692977447438238\\
};
\addlegendentry{$50\%$ Outage}

\addplot [color=black, dashed, line width=1.0pt]
  table[row sep=crcr]{%
64	69.3934124163671\\
81	35.0875222432154\\
100	20.1343694288796\\
121	12.5742815084042\\
144	8.33778710744616\\
169	5.78055701420045\\
196	4.14884311114169\\
225	3.06155634642394\\
256	2.31163184221546\\
289	1.77956300870561\\
324	1.39298049396204\\
361	1.10640223057386\\
400	0.890217828312413\\
441	0.724608465586045\\
484	0.595991511917524\\
529	0.494843996341768\\
576	0.414382396045433\\
625	0.349698410748963\\
676	0.297188370055708\\
729	0.254179623309656\\
784	0.218664512356473\\
841	0.189119209415066\\
900	0.164376165738337\\
961	0.143529951067361\\
1024	0.125871583578162\\
1089	0.110839864763831\\
1156	0\\
1225	0\\
1296	0\\
};
\addlegendentry{$1\%$ Outage}

\addplot [color=black, line width=1.0pt]
  table[row sep=crcr]{%
64	94.2216195778719\\
81	44.3559373710549\\
100	24.3080355223995\\
121	14.7172658233063\\
144	9.54772592773566\\
169	6.51370710966485\\
196	4.61786045523277\\
225	3.374568060114\\
256	2.5277501581517\\
289	1.93302460756113\\
324	1.50454329016385\\
361	1.18918333940877\\
400	0.952769477853032\\
441	0.772659938311832\\
484	0.633466781410359\\
529	0.524475486070078\\
576	0.43810436540003\\
625	0.36889935587089\\
676	0.312878861305673\\
729	0.267106706503266\\
784	0.22938889960656\\
841	0.198068671943968\\
900	0.171883008120143\\
961	0.149856293375605\\
1024	0.131227287936227\\
1089	0.115394854370168\\
1156	0\\
1225	0\\
1296	0\\
};
\addlegendentry{$0.1\%$ Outage}

\addplot [color=black, dotted, mark size=3.0pt, mark=o, mark options={solid, black}]
  table[row sep=crcr]{%
64	95\\
81	45\\
100	25\\
121	15\\
144	10\\
169	7\\
196	5\\
225	4\\
256	3\\
289	2\\
324	2\\
361	2\\
400	1\\
441	1\\
484	1\\
529	1\\
576	1\\
625	1\\
676	1\\
729	1\\
784	1\\
841	1\\
900	1\\
961	1\\
1024	1\\
1089	1\\
1156	0\\
1225	0\\
1296	0\\
};
\addlegendentry{$M_{\min}$ for $0.1\%$ Outage}

\end{axis}
\end{tikzpicture}%
	\caption{IRS-transmitter dimensional trade-off for various target outages: Here, the target rate is set to $\mathcal{R} = 3$.}
	\label{fig:outage}
\end{figure}

\section{Derivation of the Main Result}
\label{sec:derive}
This section proves the asymptotic result given in Proposition~\ref{th:0}. The derivation follows three main steps:
\begin{itemize}
	\item First, we derive the distribution of the end-to-end \ac{snr} gain $\Gamma$ for an arbitrary $M$ and $N$.
	\item We then bound the mean and variance of the \ac{snr} gain in terms of the spectrum of the covariance matrix $\mR$.
\item Finally, we send $N\rightarrow\infty$ and show that the maximum mutual information converges in distribution to a Gaussian random variable, if the maximum eigenvalue of $\mR$ grows \textit{sub-linearly} in $N$. 
\end{itemize}

\subsection{Distribution of End-to-End SNR}
Theorem~\ref{th:1} determines the statistics of the end-to-end \ac{snr} gain $\Gamma$ for an arbitrary choice of system dimensions.
\begin{theorem}
	\label{th:1}
	Let the vector of phase-shifts at the \ac{irs} be set to $\bbeta = [\beta_1,\ldots,\beta_N]^\trp$. Define $E\brc{\cdot}$ as
				\begin{align}
		E\brc{\bbeta} &=    \bar{\alpha}_N\;
		\mathbf{\bar{g}}_{\mathrm{r}}^\trp \mPhi \brc{\bbeta} {\mR} \mPhi^\her \brc{\bbeta}	\mathbf{\bar{g}}_{\mathrm{r}}^*,  \label{eq:E}
\end{align}
for matrix $\mPhi \brc{\bbeta} = \mathrm{diag} \set{ \mathrm{e}^{-\jj \beta_1},\ldots,\mathrm{e}^{-\jj \beta_N} }$, $\bar{\alpha}_N$ given in \eqref{eq:bar_alpha} and $\mathbf{\bar{g}}_{\mathrm{r}} = \mathbf{a}_N\left( \varphi_{\mathrm{r}1} , \theta_{\mathrm{r}1} \right)$. Let $F_\zeta\brc{\bbeta}$ and $\Lambda\brc{\bbeta}$ be
\begin{subequations}
			\begin{align}
		F_\zeta\brc{\bbeta} &=\bar{\alpha}_N\kappa_{\mathrm{ r} }
		\abs{ \mathbf{\bar{g}}_{\mathrm{r}}^\trp  \mPhi \brc{\bbeta} 	\mathbf{\bar{h} }_{\mathrm{r} } }^2 + \zeta E\brc{\bbeta} \label{eq:F} \\
		\Lambda\brc{\bbeta} &= \frac{\alpha_{\mathrm{ d}} A_M }{M}  + E\brc{\bbeta}. \label{eq:Lambda}
\end{align}
\end{subequations}
Then,~the end-to-end \ac{snr} gain $\Gamma$ is distributed with a generalized chi-square distribution of order $2M$ whose density is given by
\begin{align}
	f_\Gamma\brc{\gamma}= \frac{\tau_M}{ \Lambda\brc{\bbeta} } \int_{-\infty}^{+\infty} f_1\brc{\frac{\gamma - v}{ \displaystyle  \Lambda\brc{\bbeta} } } f_0 \brc{\tau_M v } \mathrm{d} v,
\end{align}
for $\tau_M ={M}/{\alpha_{\mathrm{ d}} A_M }$ and the density functions 
\begin{subequations}
	\begin{align}
	f_0\brc{v}&=\frac{v^{M-2} }{\brc{M-2}!} \mathrm{e}^{-v}, \label{eq:f_0}\\
	f_1\brc{v}&= \mathrm{e}^{-v - \lambda\brc{\bbeta}  } I_0 \brc{2 \sqrt{ \lambda\brc{\bbeta} v } } , \label{eq:f_1}
\end{align}
\end{subequations}
with $\lambda\brc{\bbeta} = {F_0\brc{\bbeta}} / {\Lambda\brc{\bbeta}}$, and denoting the modified Bessel function of the first kind and order zero. %
Moreover, the mean and variance of $\Gamma$ are given by
		\begin{align}
			\mu_\Gamma \brc{\bbeta}  &= \alpha_{\mathrm{ d}} A_M  + F_1\brc{\bbeta} ,
		\end{align}
	and \eqref{eq:sig_Gamma}, given on the top of the next page, respectively.
	\begin{figure*}[t]
 		\begin{align}
 	\sigma_\Gamma^2 \brc{\bbeta}  &=  \frac{\alpha_{\mathrm{ d}} A_M }{M} \dbc{ 2 F_1\brc{\bbeta}  + \alpha_{\mathrm{ d}} A_M }   +  E\brc{\bbeta}\dbc{F_0\brc{\bbeta}  + F_1\brc{\bbeta}} . \label{eq:sig_Gamma}
 \end{align}
\hrule
	\end{figure*}
	\end{theorem} 

\begin{proof}
	Let $\bbeta = [\beta_1,\ldots,\beta_N]^\trp$. For a given index $m\in [M]$, define $g_m\brc{\bbeta}$ be defined as
\begin{align}
	g_m \brc{\bbeta} = \sum_{n=1}^{N} \mathrm{e}^{-\jj\beta_n} h_{\mathrm{r}, n} t_{nm}.
\end{align}
Here, $\beta_n$ and  $ t_{nm}$ are deterministic scalars, and $h_{\mathrm{r}, n}$ are jointly Gaussian for $n\in [N]$. We can hence conclude that $g_m\brc{\bbeta}$ is distributed Gaussian with mean and variance
\begin{subequations}
	\begin{align}
	\mu_m \brc{\bbeta} & =\frac{ \mu_0 \brc{\bbeta}}{\sqrt{M}}
	\mathrm{e}^{ - \jj \Phi_m\left( \varphi_{\mathrm{t}2} ,\theta_{\mathrm{t}2} \right) }, \label{eq:mu_m}\\
	\tilde{\sigma}^2 \brc{\bbeta} &=  \bar{\alpha}_N\;
	\mathbf{\bar{g}}_{\mathrm{r}}^\trp \mPhi \brc{\bbeta} \mR \mPhi^\her \brc{\bbeta}	\mathbf{\bar{g}}_{\mathrm{r}}^*, 
\end{align}
\end{subequations}
respectively, where $\mathbf{\bar{g}}_{\mathrm{r}} = \mathbf{a}_N\left( \varphi_{\mathrm{r}1} , \theta_{\mathrm{r}1} \right)$, $\bar{\alpha}_N$ is defined in \eqref{eq:bar_alpha},
	\begin{align}
	\mu_0 \brc{\bbeta} & =\sqrt{ M \kappa_\mathrm{r} \bar{\alpha}_N }
	\mathbf{\bar{g}}_{\mathrm{r}}^\trp \mPhi \brc{\bbeta}	\mathbf{\bar{h}}_{\mathrm{r}},
\end{align}
and 
	$\mPhi \brc{\bbeta} = \mathrm{diag} \set{ \mathrm{e}^{-\jj \beta_1},\ldots,\mathrm{e}^{-\jj \beta_N} }$.

As $h_m$ for $m\in\dbc{M}$ have common term $g_m\brc{\bbeta}$, the entries of $\bh$ are correlated. We hence determine the covariance of $\bh$. To this end, we first note that 
\begin{align}
	\bmu\brc{\bbeta} = \dbc{\mu_1 \brc{\bbeta}, \ldots , \mu_M \brc{\bbeta}}^\trp = \mu_0 \brc{\bbeta}  \bv^*
\end{align}
where $\bv = \mathbf{a}_M\left( \varphi_{\mathrm{t}2} , \theta_{\mathrm{t}2} \right)/\sqrt{M}$. As the result, the covariance matrix is given by
\begin{subequations}
	\begin{align}
	\mC &= \Ex{ \brc{\bh - \mu_0 \brc{\bbeta} \bv^*} \brc{\bh^\her - \mu_0^* \brc{\bbeta} \bv^\trp}},\\
	&= \alpha_{\mathrm{ d}} A_M  \mI_M + M \tilde{\sigma}^2 \brc{\bbeta} \bv^* \bv^\trp = \mC_0^2,
\end{align}
\end{subequations}
where $\mC_0 = \sqrt{\alpha_{\mathrm{ d}} A_M } \; \mI_M + \varpi \brc{\bbeta} \bv^* \bv^\trp$ with $\varpi$ being 
	\begin{align}
\varpi \brc{\bbeta} =    \sqrt{\alpha_{\mathrm{ d}} A_M + M {\tilde{\sigma}^2 \brc{\bbeta} } } - \sqrt{\alpha_{\mathrm{ d}} A_M }.
\end{align}

We now represent the end-to-end channel $\bh$ as $\bh = \mC_0 \bh_0$, where $\bh_0 \sim \mathcal{CN}\left( \mC_0^{-1} \bmu\brc{\bbeta}, \mI_M\right)$. Consequently, we have
\begin{align}
\Gamma = \frac{1}{M} \bh_0^\her \mC \bh_0. \label{eq:Gam}
\end{align}
To determine this weighted norm, we use the~eigenvalue~decomposition of $\mC$ to write $\mC = \mV^* \mSigma \mV^\trp$, where $\mV\in\setC^{M\times M}$ is a unitary matrix whose \textit{last} column is $\bv$ and $\mSigma $ is an $M\times M$ diagonal matrix defined as
\begin{align}
	\mSigma = \mathrm{diag} \set{ \alpha_{\mathrm{ d}} A_M , \ldots,\alpha_{\mathrm{ d}} A_M ,\alpha_{\mathrm{ d}} A_M +M  \tilde{\sigma}^2 \brc{\bbeta} }.
\end{align}
By replacing into \eqref{eq:Gam}, we have
\begin{align}
	\Gamma = \frac{\alpha_{\mathrm{ d}} A_M }{M}  \sum_{m=1}^{M-1} \abs{r_m}^2 + \frac{\alpha_{\mathrm{ d}}A_M +M \tilde{\sigma}^2 \brc{\bbeta} }{M}   \abs{r_M}^2 \label{eq:Gamma_1}
\end{align}
where $r_m$ is the $m$-th entry of $\mathbf{r} = \mV^\trp \bh_0$. Noting that $\mV^\trp$ is unitary, we conclude that $\mathbf{r}\sim \mathcal{CN}\left(   \bmu_{\mathrm{r}}\brc{\bbeta}, \mI_M\right)$, where
\begin{align}
	\bmu_{\mathrm{r}}\brc{\bbeta} = \mV^\trp\mC_0^{-1} \bmu \brc{\bbeta}.
\end{align}

Since $\mC = \mC_0^{2}$, the matrix $\mC_0^{-1}$ is decomposed as 
\begin{align}
	\mC_0^{-1} = \mV^* \sqrt{\mSigma}^{\; -1} \mV^\trp,
\end{align}
where $\sqrt{\mSigma}$ is a diagonal matrix whose first $M-1$ diagonal entries is $\sqrt{\alpha_{\mathrm{ d}} A_M }$ and the last one is $\sqrt{\alpha_{\mathrm{ d}} A_M }+\varpi \brc{\bbeta}$. %
Hence, we can write
\begin{subequations}
	\begin{align}
	\bmu_{\mathrm{r}}\brc{\bbeta} &= \mV^\trp \mV^* \sqrt{\mSigma}^{\; -1} \mV^\trp \bmu \brc{\bbeta}\\
	&= \mu_0 \brc{\bbeta}  \sqrt{\mSigma}^{\; -1} \mV^\trp \bv^*.
\end{align}
\end{subequations}
Noting that $\bv$ is the last column of the \textit{unitary} matrix $\mV$, we conclude that $\mV^\trp \bv^*$ is a standard base vector with the first $M-1$ being zero and the last entry being $1$. This concludes that $\mu_{\mathrm{r},m}\brc{\bbeta} = 0$ for $m\in \dbc{M-1}$ and 
\begin{subequations}
	\begin{align}
	\mu_{\mathrm{r},M}\brc{\bbeta} &= \frac{\mu_0 \brc{\bbeta} }{\sqrt{\alpha_{\mathrm{ d}}A_M }+\varpi \brc{\bbeta} }\\
	&= \frac{\mu_0 \brc{\bbeta} }{\sqrt{\alpha_{\mathrm{ d}} A_M + M {\tilde{\sigma}^2 \brc{\bbeta} } }  } \label{eq:mu_r_M}
\end{align}
\end{subequations}

We now define \textit{independent} random variables
\begin{subequations}
	\begin{align}
	V_0 &=   \sum_{m=1}^{M-1} \abs{r_m}^2, \label{eq:v0}\\
	V_1 &=    \abs{r_M}^2. \label{eq:v1}
\end{align}
\end{subequations}
By definition, $V_0$ is distributed chi-square with $2M-2$ degrees for freedom whose mean and variance are $M-1$ and whose density function is given by \eqref{eq:f_0}. %
The random variable $V_1$ is further distributed non-central chi-square with two degrees of freedom and non-centrality parameter 
\begin{align}
	\lambda\brc{\bbeta} = \abs{\mu_{\mathrm{r},M}\brc{\bbeta}}^2.
\end{align}
Consequently, the mean and variance of $V_1$ are $1+\lambda\brc{\bbeta}$ and $1+ 2 \lambda\brc{\bbeta}$, respectively, and its probability density is given by \eqref{eq:f_1} in Theorem~\ref{th:1}.

The \ac{snr} gain $\Gamma$ is determined in terms of $V_0$ and $V_1$ as
\begin{align}
	\Gamma = \frac{\alpha_{\mathrm{ d}} A_M }{M}  V_0 + \frac{\alpha_{\mathrm{ d}} A_M +M \tilde{\sigma}^2 \brc{\bbeta} }{M}   V_1. \label{eq:Gamma_sum}
\end{align}
Noting that $V_0$ and $V_1$ are independent, the proof of Theorem~\ref{th:1} is concluded after defining the functions $E\brc{\bbeta} = \tilde{\sigma}^2\brc{\bbeta}$ and $F_\zeta\brc{\bbeta} = \abs{\mu_0 \brc{\bbeta}}^2/M + \zeta E\brc{\bbeta}$.
\end{proof} 

\subsection{Maximum Average End-to-End SNR}
Using \acp{irs} in the system, we are often interested in tuning the \ac{irs} phase-shifts, such that the average end-to-end \ac{snr} is maximized, i.e., $\mu_\Gamma \brc{\bbeta}$ is maximized over $\bbeta$. Considering Theorem~\ref{th:1}, the maximum mean \ac{snr} is not necessarily found in a closed-form. Nevertheless, closed-form lower and upper bounds can be derived in terms of the spectrum of $\mR$.

\begin{theorem}
	\label{th:2}
	Let the phase-shifts at the  \ac{irs} be set to the vector $ \bbeta^\star = [\beta^\star_1,\ldots,\beta^\star_N]^\trp$ with $\beta_n^\star$ being specified in \eqref{eq:beta_n_star}. %
	Then, the average end-to-end \ac{snr} gain is bounded as
\begin{subequations}
			\begin{align}
			  \mu_\Gamma \brc{\bbeta^\star} &\leq  \xi_{\max} \brc{N} +  \kappa_{\mathrm{r}} \bar{\alpha}_N N^2 , \\
			 \mu_\Gamma \brc{\bbeta^\star} &\geq \xi_{\min} \brc{N} +  \kappa_{\mathrm{r}} \bar{\alpha}_N N^2  
		\end{align}
\end{subequations}
	where $\xi_i \brc{N}$ for $i\in\set{\min , \max}$ is defined as
	\begin{align}
		\xi_i \brc{N} = {\alpha_{\mathrm{ d}} A_M+ \bar{\alpha}_N\lambda_i N } \label{eq:xi_i}
	\end{align}
with $\bar{\alpha}_N$ being given in \eqref{eq:bar_alpha}, and $\lambda_{\min}$ and $\lambda_{\max}$ denoting the minimum and maximum eigenvalue of $\mR$, respectively. The variance of $\Gamma$ is further bounded as 
	\begin{align}
	P_{\min}\brc{N} \leq \sigma_\Gamma^2 \brc{\bbeta^\star} &\leq P_{\max}\brc{N},
\end{align}
where $P_{i}\brc{N}$ is given for $i\in\set{\min,\max}$ in \eqref{eq:upp_sig} at the top of the next page.
\begin{figure*}[t]
	\begin{align}
		P_{i}\brc{N} =  \frac{ \alpha_{\mathrm{ d}} A_M }{M} +  \frac{2\alpha_{\mathrm{ d}} A_M \lambda_{i} }{M} \bar{\alpha}_N N
		+ \brc{ \frac{2\alpha_{\mathrm{ d}} A_M \kappa_{\rm r} }{M} \bar{\alpha}_N +\bar{\alpha}_N^2 \lambda_{i}^2 } N^2
		+ 2 \lambda_{i} \kappa_{\mathrm{ r} } \bar{\alpha}_N^2 N^3 \label{eq:upp_sig}
	\end{align}
\hrule
\end{figure*}
\end{theorem}
\begin{proof}
	We start the proof by stating the following lemma:
\begin{lemma}
		\label{lem:1}
	For the function $F_0\brc{\bbeta}$ defined in \eqref{eq:F}, we have
	\begin{align}
		 \max_{\beta_1,\ldots,\beta_N} F_0 \brc{\bbeta}  = \kappa_{\mathrm{r}} \bar{\alpha}_N N^2,
	\end{align}
being obtained by setting $\beta_n = \beta^\star_n$ given in \eqref{eq:beta_n_star} for $n\in \dbc{N}$. 
\end{lemma}
\begin{proof}
	The proof is given in the Appendix~\ref{app:A}.
\end{proof}
By setting $\beta_n = \beta^\star_n$ in Theorem~\ref{th:1}, the mean and variance of the \ac{snr} gain reduce to
	\begin{align}
	\mu_\Gamma \brc{\bbeta^\star} &=  \alpha_{\mathrm{ d}} A_M+  E\brc{\bbeta^\star} + \kappa_{\mathrm{r}}  \bar{\alpha}_N N^2 , 
\end{align}
and the expression given in \eqref{eq:sig_star} at the top of the next page, respectively.
\begin{figure*}[t]
			\begin{align}
		\sigma_\Gamma^2 \brc{\bbeta^\star}  &=  \frac{\alpha_{\mathrm{ d}} A_M }{M} \dbc{ 2\kappa_{\mathrm{r}}  \bar{\alpha}_N N^2+ 2 E\brc{\bbeta^\star} + \alpha_{\mathrm{ d}} A_M }  +  E\brc{\bbeta^\star}\dbc{ E\brc{\bbeta^\star} + 2 \kappa_{\mathrm{r}}  \bar{\alpha}_N N^2} , \label{eq:sig_star}
		\end{align}
\hrule
		\begin{align}
		\sigma_{\Gamma,i}^2   &=  \frac{\alpha_{\mathrm{ d}} A_M }{M} \brc{ 2\kappa_{\mathrm{r}}  \bar{\alpha}_N N^2+ 2 \maE_{i} - \alpha_{\mathrm{ d}} A_M }  +  \brc{\maE_{i} - \alpha_{\mathrm{ d}} A_M } \brc{ \maE_{i} + 2 \kappa_{\mathrm{r}}  \bar{\alpha}_N N^2 - \alpha_{\mathrm{ d}} A_M } , \label{eq:sig_max}
	\end{align}
\hrule
\end{figure*}
The mean can hence be bounded as
	\begin{align}
\maE_{\min} + \kappa_{\mathrm{r}}  \bar{\alpha}_N N^2 \leq	\mu_\Gamma \brc{\bbeta^\star} \leq \maE_{\max} + \kappa_{\mathrm{r}}  \bar{\alpha}_N N^2, \label{eq:Mu_B}
\end{align}
and the variance reads
\begin{align}
	\sigma_{\Gamma,\min}^2 \leq \sigma_\Gamma^2 \brc{\bbeta^\star} &\leq  \sigma_{\Gamma,\max}^2, \label{eq:SigB} 
\end{align}
with $\sigma_{\Gamma,i}^2$ being defined in \eqref{eq:sig_max} at the top of the next page for $i\in \set{\min,\max}$. Here, 
\begin{align}
	\maE_{\max} &= \alpha_{\mathrm{ d}} A_M+  \max_{\bbeta\in\setR^N} E\brc{\bbeta} \\
\maE_{\min} &= \alpha_{\mathrm{ d}} A_M+  \min_{\bbeta\in\setR^N} E\brc{\bbeta}
\end{align}
whose values are given by the following lemma:
\begin{lemma}
	\label{lem:2}
	For the function $E\brc{\bbeta}$ defined in \eqref{eq:E}, we have
	\begin{align}
	\xi_{\min} \brc{N} \leq \alpha_{\mathrm{ d}} A_M + E \brc{\bbeta}  \leq \xi_{\max} \brc{N}
	\end{align}
for any $\bbeta\in\setR^N$, where the function $\xi_i\brc{N}$ is defined in \eqref{eq:xi_i} for $i\in \set{\min,\max}$.
\end{lemma}
\begin{proof}
	The proof is given in the Appendix~\ref{app:B}.
\end{proof}
By substituting the results of Lemma~\ref{lem:2} into the bounds in \eqref{eq:SigB} and \eqref{eq:Mu_B}, the proof is concluded.
%
%
%
\end{proof}
\begin{remark}
From the proof, it is straightforward to conclude that the upper-bounds given in Lemma~\ref{lem:2} are in general valid for any choice of \ac{irs} phase-shifts. This means that, for any $\bbeta\in\setR^N$, we have
\begin{subequations}
			\begin{align}
\mu_\Gamma \brc{\bbeta} &\leq  \xi_{\max} \brc{N} +  \kappa_{\mathrm{r}}  \bar{\alpha}_N N^2 , \\
\sigma_\Gamma^2 \brc{\bbeta} &\leq P_{\max}\brc{N}.
\end{align}
\end{subequations}
	\end{remark}
\begin{remark}
Lemma~\ref{lem:1} explains the logic behind the choice~of \ac{irs} phase-shifts in Proposition~\ref{th:0}. It is worth noting that the proposed phase-shifts in \eqref{eq:beta_n_star} do not necessarily maximize the average end-to-end \ac{snr}. Nevertheless, they are still good enough to guarantee the end-to-end channel hardening. 
\end{remark}

\subsection{Proof of Proposition~\ref{th:0}}
\label{proof:Prop}
The proof of Proposition~\ref{th:0} follows readily from Theorems~\ref{th:1} and \ref{th:2}: Using Theorem~\ref{th:1}, we represent the end-to-end \ac{snr} gain as $\Gamma = \mu_\Gamma \brc{\bbeta^\star}+\sigma_\Gamma \brc{\bbeta^\star} \tilde{\Gamma}$, where $\tilde{\Gamma}$ is a zero-mean and unit-variance generalized chi-square random variable denoting the centralized and normalized version of $\Gamma$. We now replace it in \eqref{eq:I} and use the Taylor series of the logarithm to write
\begin{subequations}
	\begin{align}
		\mathcal{C} &=\log_2 \left(1+  \rho M \mu_\Gamma\brc{\bbeta^\star}+ \rho M \sigma_\Gamma\brc{\bbeta^\star} \tilde{\Gamma} \right),\\
		&= \log_2 \left(1 +  \rho M \mu_\Gamma\brc{\bbeta^\star} \right)  + \log_2 \left(1+  \varsigma_N \tilde{\Gamma} \right)\\
		&= \log_2 \left(1 +\rho M \mu_\Gamma\brc{\bbeta^\star}  \right)  +  \varsigma_N \log_2\mathrm{e}  \; \tilde{\Gamma} + \epsilon_N, \label{eq:taylor}
	\end{align}
\end{subequations}
where $\varsigma_N$ is defined as\footnote{We use subscript $N$ to indicate its dependency on dimension $N$.}
\begin{align}
	\varsigma_N = \frac{ \rho M \sigma_\Gamma \brc{\bbeta^\star}   }{1 + \rho M \mu_\Gamma \brc{\bbeta^\star}  },
\end{align}
and $ \epsilon_N$ is a polynomial in $\varsigma_N\tilde{\Gamma}$ satisfying
\begin{align}
	\epsilon_N = \mathcal{O}\left(  {\varsigma_N^2 \tilde{\Gamma}^2} \right).
\end{align}
Invoking Theorem~\ref{th:2}, we can bound $\varsigma_N$ from above as
\begin{subequations}
	\begin{align}
	\varsigma_N &\leq \frac{ \rho M \sqrt{ P_{\max} \brc{N} }  }{1 + \rho M \brc{ \xi_{\min} \brc{N} +  \kappa_{\mathrm{r}}  \bar{\alpha}_N N^2} }\\
	&\stackrel{\dagger}{\leq} \frac{ \rho M \sqrt{ P_{\max} \brc{N} }  }{1 + \rho M \brc{ \alpha_{\mathrm{ d}} A_M  +  \kappa_{\mathrm{r}}  \bar{\alpha}_N N^2} }
\end{align}
\end{subequations}
where $\dagger$ follows from $\xi_{\min} \brc{N}\geq \alpha_{\mathrm{ d}} A_M $.

We now focus on scaling with $N$. Since $\lambda_{\max}$ grows \textit{sub-linearly} in $N$, we can conclude that there exist real scalars $\lambda_0$ and  $0\leq u <1$, such that $\lambda_{\max}$ is uniformly bounded by 
	\begin{align}
	\lambda^{\rm U} \brc{N} &= \lambda_0 N^{u},
\end{align}
From the constraint $A_{\rm IRS} = A_0 N^{1-q}$, we can further conclude that
	\begin{align}
	A_N &= \frac{A_0}{N^q},
\end{align}
for some real $A_0$ and $0\leq q <1$, where the constraint~in~\eqref{eq:const_IRS} guarantees that $u\geq q$. As a result, $P_{\max} \brc{N}$ is uniformly bounded from above by
	\begin{align}
		P^{\rm U}\brc{N} &= b N^{3+u - 4q}.
	\end{align}
These upper-bounds lead to this conclusion that $\varsigma_N$ is bounded as $\varsigma_N \leq \varsigma_N^{\rm U}$ for some $\varsigma_N^{\rm U}$ satisfying
\begin{align}
	\varsigma_N^{\rm U} = \mathcal{O}\brc{N^{\tfrac{u-1}{2}}}.
\end{align}
We now note that $\tilde{\Gamma}$ is zero-mean and unit-variance, and $0\leq u < 1$. Thus, $\epsilon_N \rightarrow 0$ as $N$ grows large, and hence $\mathcal{C}$ is well approximated by the first two terms in \eqref{eq:taylor}.

We now use the following lemma:
\begin{lemma}
	\label{lem:3}
	Let $F_0\brc{\bbeta}/\Lambda\brc{\bbeta}$ and $\sigma^2_\Gamma\brc{\bbeta}$ grow large. Then, we have
	\begin{align}\label{eq:V}
		\dfrac{\Gamma-\mu_\Gamma \brc{\bbeta}}{\sigma_\Gamma\brc{\bbeta} }\xrightarrow{\mathrm{ d}} \mathcal{N}\left( 0, 1\right).
	\end{align} 
\end{lemma}
\begin{proof}
The proof is given in the Appendix~\ref{app:3}.
\end{proof}

We first note that $\sigma^2_\Gamma\brc{\bbeta^\star}$ grows large with $N$. Moreover, $F_0\brc{\bbeta^\star} = \bar{\alpha}_N\kappa_{\mathrm{ r} } N^2$ and $\Lambda\brc{\bbeta^\star}$ is uniformly bounded from above with $ \mathcal{O}\brc{N^{1+u - 2q}}$ for some $0\leq u <1$. This concludes that $F_0\brc{\bbeta^\star}/\Lambda\brc{\bbeta^\star}$ grows large as $N$ increases. We hence use Lemma~\ref{lem:3} and conclude that $\tilde{\Gamma}$ converges in distribution to a zero-mean and unit-variance Gaussian random variable, as $N$ grows asymptotically large.

Finally by noting that $\bar{\mathbf{g}}_{\mathrm{ r}}^\trp \mPhi\brc{\bbeta^\star} = \mathbf{\bar{h}}_{\mathrm{r}}^\her$, we can write
  \begin{align}
E\brc{\bbeta^\star}  =  \bar{\alpha}_N\mathbf{\bar{h}}_{\mathrm{r}}^\her \mR \mathbf{\bar{h}}_{\mathrm{r}}.
 \end{align}
Substituting into \eqref{eq:taylor}, Proposition~\ref{th:0} is concluded after few lines of standard calculations.

\subsection{Derivation of Pseudo-Proposition~\ref{th:0-2}}
\label{sec:Proof_PPro}
From our numerical validations in Section~\ref{sec:Linear} we heuristically conclude that $\mathcal{C}$ is still well-approximated by a Gaussian random variable in the linearly scaling settings. Nevertheless, the variance of $\mathcal{C}$ in this case is not tightly approximated by Proposition~\ref{th:0}. From the proof of Proposition~\ref{th:0}, the inaccuracy comes from the fact that the necessary condition for Lemma~\ref{lem:3} is not satisfied and hence, $\tilde{\Gamma}$ does not converge in distribution to a zero-mean and unit-variance Gaussian random variable\footnote{This is why Proposition~\ref{th:0} excludes the strictly-linear case.}. 

The exact distribution of $\tilde{\Gamma}$ in this case can~be~derived~explicitly from Theorem~\ref{th:1}. We however follow our heuristic and approximate it by a zero-mean Gaussian random variable of a smaller variance. Following the proof of Lemma~\ref{lem:3}, we show in Appendix~\ref{app:4} that $\tilde{\Gamma}$ can be approximated in this case as 
\begin{align}
	\tilde{\Gamma} \approx \tilde{\Gamma}_\infty + \hat{\epsilon} \label{eq:approx_Gam}
\end{align}
where $\tilde{\Gamma}_\infty\sim \mathcal{N}\brc{0, {\sigma}_{\infty}^2 }$ with 
\begin{align}
	{\sigma}_{\infty}^2 = \frac{\kappa_{\mathrm{ r} }}{\kappa_{\mathrm{ r} } + \vartheta } \label{eq:sig_1_final}
\end{align}
and $\vartheta$ being given in \eqref{eq:vartheta}. The random variable $\hat{\epsilon}$ is further a centralized\footnote{That means $\hat{\epsilon}$ is zero-mean.} chi-square residual term. By replacing \eqref{eq:approx_Gam} in the Taylor expansion \eqref{eq:taylor} and ignoring the residual term as well as higher-order terms, the approximation is derived.

\section{Conclusions}
\label{sec:conc}
Even with highly correlated reflecting elements, a~large~\ac{irs} hardens the end-to-end channel of a \ac{mimo} system with~small transmit and receive array antennas. This conclusion follows from the main result of this study which shows that the mutual information between the input and output signals in an \ac{irs}-aided system is almost perfectly approximated by a Gaussian random variable, whose mean increases with the \ac{irs} size and whose variance vanishes to zero as the \ac{irs} size grows large. 

The results of this study indicate that the channel hardening property is rather generic in \ac{irs}-aided systems and is easily achieved by small transmit and receive array antennas. More interestingly, the phase-shifts of reflecting elements are not required to be updated frequently and should only be appropriately matched with the large-scale statistics of the wireless channels, i.e., \acp{aoa} and\acp{aod}. These findings reveal this interesting fact: Even with \acp{irs} of moderate size whose physical dimensions are rather significantly smaller than the wavelength, we can still harden the end-to-end \ac{mimo} channel between the transmitter and the receiver. 

Earlier investigations in the literature have demonstrated the promising performance of massive \ac{mimo} systems in various respects. Nevertheless, the complexity of implementing such systems left this technology challenging for commercial uses.  Along with earlier results in the literature, the result of this study indicate that using \acp{irs}, the key features of the massive \ac{mimo} technology can be achieved with rather small end-to-end dimensions. 

Although this work studies a basic setting,~the~analytical results can still be employed to investigate various fundamental properties of \ac{irs}-aided \ac{mimo} systems. An example is given in Section~\ref{sec:trade}, where we demonstrate the dimensional trade-off between the \ac{irs} and the \ac{bs}: The larger the \ac{irs} is, the smaller the \ac{bs} needs to be in order to achieve a target performance. Similar investigations and further extensions of the results to more advanced settings, and under more realistic assumptions, are interesting study directions. They are however out of the scope of this particular paper and are left as potential directions for future work.

\appendices
\section{Proof of Lemma~\ref{lem:1}}
\label{app:A}
Considering the definition of $F_0\brc{\bbeta}$, one can write
			\begin{align}
	F_0\brc{\bbeta} &=\bar{\alpha}_N\kappa_{\mathrm{ r} }
	\abs{
\sum_{n=1}^{N} 	\mathrm{e}^{
\jj \brc{\Pi_n\left( \varphi_{\mathrm{r}1} ,\theta_{\mathrm{r}1} \right)  + \Pi_n\left( \varphi_{\mathrm{t}1} ,\theta_{\mathrm{t}1} \right) - \beta_n}
}
}^2.
\end{align}
The expression on the right hand side includes the amplitude of a sum whose summands are complex numbers on the unit circle. As the result, the amplitude of the sum is maximized by setting all the summands in-phase. This is obtained by setting $\beta_n = \beta_n^\star$ for $\beta_n^\star$ given in \eqref{eq:beta_n_star}. In this case, the sum adds to $N$ and $F\brc{\bbeta^\star}$ is given by the expression given in Lemma~\ref{lem:1}.

\section{Proof of Lemma~\ref{lem:2}}
\label{app:B}
Let $\btau\brc{\bbeta}  = \mPhi^\her\brc{\bbeta} \mathbf{g}_{\mathrm{r}}^*$. Hence, $E\brc{\bbeta}$ can be written as
			\begin{align}
	E\brc{\bbeta} &= \bar{\alpha}_N\;
	\btau^\her\brc{\bbeta}  {\mR} \btau \brc{\bbeta} \label{eq:e_neu}
\end{align}
Since every entry of $\btau\brc{\bbeta}$ lies on the unit circle, we can write $\norm{\btau\brc{\bbeta}}^2 = N$ for any $\bbeta\in\setR^N$. Consequently, the function
\begin{align}
	Q\brc{\bbeta} = \frac{1}{N} \btau^\her\brc{\bbeta} \mR \btau\brc{\bbeta}
\end{align}
determines the \textit{Rayleigh quotient} of $\mR$ at $\btau\brc{\bbeta}$. We now use the fact that the {Rayleigh quotient} is bounded in terms of the eigenvalue of $\mR$, as follows \cite{horn2013matrix}
\begin{align}
	\lambda_{\min} \leq Q\brc{\bbeta} \leq \lambda_{\max},
\end{align}
where $\lambda_{\min}$ and $\lambda_{\max}$ are the minimum and maximum eigenvalue of $\mR$, respectively. Substituting into \eqref{eq:e_neu}, Lemma~\ref{lem:2} is concluded.
%

\section{Proof of Lemma~\ref{lem:3}}
\label{app:3}
We start the proof by considering \eqref{eq:Gamma_1}. Using the definition of  $\Lambda\brc{\bbeta}$ in Theorem~\ref{th:1}, we have
\begin{align}
	\Gamma = \frac{\alpha_{\mathrm{ d}} A_M }{M}  \sum_{m=1}^{M-1} \abs{r_m}^2 + \Lambda\brc{\bbeta}  \abs{r_M}^2.
\end{align}
Here, $r_m$ for $m\in\dbc{M-1}$ are independent complex Gaussian random variables with zero mean and unit variance,~and~$r_M$ is a complex unit-variance Gaussian random variable independent of $r_m$ for $m\in\dbc{M-1}$. The mean of $r_M$, i.e., $\mu_{\mathrm{r},M}\brc{\bbeta}$, is given in \eqref{eq:mu_r_M}. Using the definitions given in Theorem~\ref{th:1}, it can be shown that
	\begin{align}
	\abs{\mu_{\mathrm{r},M}\brc{\bbeta}} = \sqrt{\frac{F_0 \brc{\bbeta} }{\Lambda\brc{\bbeta}} }.
\end{align}

Let $r_M = \tilde{r}_M + \mu_{\mathrm{r},M}\brc{\bbeta}$ where $ \tilde{r}_M$ is the centralized form of $r_M$. We can hence write
\begin{align}
	\Gamma = T_0  + 2 \Lambda\brc{\bbeta}  \Re\set{\tilde{r}_M \mu_{\mathrm{r},M}^*\brc{\bbeta} } + \abs{\mu_{\mathrm{r},M}\brc{\bbeta}}^2 \label{eq:Gamma_T}
\end{align}
where $T_0$ is defined as
\begin{align}
	T_0 = \frac{\alpha_{\mathrm{ d}} A_M }{M}  \sum_{m=1}^{M-1} \abs{r_m}^2 + \Lambda\brc{\bbeta}  \abs{\tilde{r}_M}^2.
\end{align}

We now consider the normalized \ac{snr} gain, i.e., 
\begin{align}
	\hat{\Gamma} &= \frac{\Gamma}{\sigma_\Gamma\brc{\bbeta}}
\end{align}
By replacing $\Gamma$ with the expression in \eqref{eq:Gamma_T}, we have
\begin{align}
	\hat{\Gamma} &= \Gamma_0 + \Gamma_1 + c \label{eq:Decomp}
\end{align}
where $\Gamma_0$ and $\Gamma_1$ are random expressions defined as
\begin{subequations}
	\begin{align}
	\Gamma_0 &= \frac{T_0}{\sigma_\Gamma\brc{\bbeta}}  \\
	\Gamma_1 &= 2 \frac{\Lambda\brc{\bbeta}}{\sigma_\Gamma\brc{\bbeta}}  \Re\set{\tilde{r}_M \mu_{\mathrm{r},M}^*\brc{\bbeta} }
\end{align}
\end{subequations}
and $c$ is a deterministic constant given by
\begin{align}
	c = \frac{\abs{\mu_{\mathrm{r},M}\brc{\bbeta}}^2}{\sigma_\Gamma\brc{\bbeta}}.
\end{align}
Despite its complicated form of $\sigma^2_\Gamma\brc{\bbeta}$, we can use the fact that $F_0\brc{\bbeta} \leq F_1 \brc{\bbeta}$ in Theorem~\ref{th:1}, and bound $\sigma^2_\Gamma\brc{\bbeta}$ as
\begin{align}
	2 \Lambda\brc{\bbeta} F_0\brc{\bbeta} \leq \sigma^2_\Gamma\brc{\bbeta} \leq 2 \Lambda\brc{\bbeta} F_1\brc{\bbeta} + \alpha_{\mathrm{ d}}^2 A_M^2 .
\end{align}	
Noting that $\alpha_{\mathrm{ d}}^2 A_M^2$ is fixed, we can conclude that growth of $\sigma^2_\Gamma\brc{\bbeta}$ guarantees that $ \Lambda\brc{\bbeta} F_0\brc{\bbeta}$ grows large.

We use the upper bound to write
\begin{subequations}
	\begin{align}
	\Gamma_0   &\leq  \frac{T_0}{2\sqrt{\Lambda \brc{\bbeta} F_0\brc{\bbeta} }}\\
	&= \frac{\displaystyle \alpha_{\mathrm{ d}} A_M \sum_{m=1}^{M-1} \abs{r_m}^2 }{2M \sqrt{\Lambda \brc{\bbeta} F_0\brc{\bbeta}}}   + \sqrt{ \frac{ \Lambda\brc{\bbeta} }{ F_0\brc{\bbeta} } } \abs{\tilde{r}_M}^2 .
\end{align}
\end{subequations}
Now, let $\sigma^2_\Gamma\brc{\bbeta}$  and $ F_0\brc{\bbeta} /  \Lambda\brc{\bbeta}$ grow large. We can then conclude that the upper bound in this case converges to zero in the \ac{mse}. Since $\Gamma_0 \geq 0$, we can conclude that $\Gamma_0$ converges to zero.

The expression $\Gamma_1$ is further a real-valued zero-mean Gaussian random variable whose variance is given by
\begin{align}
	\sigma_1^2 &=  \frac{2\Lambda^2\brc{\bbeta}}{\sigma_\Gamma^2\brc{\bbeta}}  \abs { \mu_{\mathrm{r},M}\brc{\bbeta} }^2. \label{sig_1}
\end{align}
Using the lower and upper bounds on the variance of $\Gamma$, we can bound $\sigma_1^2$ as
\begin{align}
	\frac{F_0\brc{\bbeta}}{F_1\brc{\bbeta}} \leq \sigma_1^2 \leq 1.
\end{align}
Noting that 
\begin{align}
	F_1\brc{\bbeta} = F_0\brc{\bbeta} + \Lambda \brc{\bbeta} - \frac{\alpha_{\mathrm{ d}} A_M }{M}
\end{align}
we can further write
\begin{align}
\lim_{\tfrac{F_0\brc{\bbeta}}{\Lambda\brc{\bbeta}} \rightarrow \infty }	\frac{F_0\brc{\bbeta}}{F_1\brc{\bbeta}} = 1
\end{align}
which means that $\sigma_1^2 = 1$.  As a result, by growth of $\sigma^2_\Gamma\brc{\bbeta}$ and $F_0\brc{\bbeta} /  \Lambda\brc{\bbeta}$, the normalized \ac{snr} gain $\hat{\Gamma}$ converges in distribution to a unit-variance real Gaussian random variable. 

Finally, by noting that $\brc{\Gamma - \mu_\Gamma} / \sigma_\Gamma$ is the centralized form of $\hat{\Gamma}$, we can conclude that 
	\begin{align}
	\dfrac{\Gamma-\mu_\Gamma \brc{\bbeta}}{\sigma_\Gamma\brc{\bbeta} }\xrightarrow{\mathrm{ d}} \mathcal{N}\left( 0, 1\right).
\end{align}
This concludes the proof. 

\section{Approximated SNR for Linear Scaling}
\label{app:4}
For a linearly scaling \ac{irs} area and $\lambda_{\max}$, the conditions in Lemma~\ref{lem:3} are not satisfied, and hence the given limit is not valid. Considering the proof in Appendix~\ref{app:3}, this follows from the fact that due to linear scaling $\Gamma_0$ does not converge to zero, and accordingly the limit of $\sigma_1^2$ is not one. Considering \eqref{eq:Decomp} in Appendix~\ref{app:3}, we can write
\begin{align}
	\tilde{\Gamma} &= \hat{\Gamma} - \frac{\mu_\Gamma\brc{\bbeta} }{\sigma_\Gamma\brc{\bbeta}} = \Gamma_1 + \hat{\epsilon},
\end{align}
where $\hat{\epsilon}$ is a centralized chi-square random variable\footnote{$\hat{\epsilon}$ is zero-mean, since $\tilde{\Gamma}$ is zero-mean}.~The~random variable $\Gamma_1$ is further zero-mean Gaussian whose variance is $\sigma_1^2$ given in \eqref{sig_1}.

Considering linear scaling, i.e., $q=u = 1$, $\mR$ is a rank-one matrix and hence we can write $\mR = \mathbf{r} \mathbf{r}^\her$ for some $\mathbf{r}\in\setC^N$. As a result,
\begin{align}
	\mathbf{\bar{h}}_{\mathrm{r}}^\her \mR \mathbf{\bar{h}}_{\mathrm{r}} = \abs{\mathbf{h}^\her \mathbf{r} }^2.
\end{align}
Noting that entries of $\mathbf{h}^\her$ lie equidistantly on the unit-circle, one can approximately write $\mathbf{\bar{h}}_{\mathrm{r}}^\her \mR \mathbf{\bar{h}}_{\mathrm{r}} \approx 0$ for a typical\footnote{For instance, for the covariance matrix $\mR = \boldsymbol{1}_N$, $\mathbf{r}$ is a vector of all-ones, and hence this approximation is almost exact.} $\mathbf{r}$.

We further note that $\bar{\alpha}_N N^2$ converges to a constant with $q=1$. Considering this fact and using the above approximation, the variance $\sigma_1^2$ is approximated with $\sigma_{\infty}^2$  given in \eqref{eq:sig_1_final}. 

\begin{acronym}
\acro{mimo}[MIMO]{multiple-input multiple-output}
\acro{miso}[MISO]{multiple-input single-output}
\acro{siso}[SISO]{single-input single-output}
\acro{csi}[CSI]{channel state information}
\acro{awgn}[AWGN]{additive white Gaussian noise}
\acro{iid}[i.i.d.]{independent and identically distributed}
\acro{ut}[UT]{user terminal}
\acro{bs}[BS]{base station}
\acro{snr}[SNR]{signal-to-noise ratio}
\acro{rf}[RF]{radio frequency}
\acro{los}[LoS]{line of sight}
\acro{nlos}[NLoS]{non-line of sight}
\acro{irs}[IRS]{intelligent reflecting surface}
\acro{ula}[ULA]{uniform linear array}
\acro{aoa}[AoA]{angle-of-arrival}
\acro{aod}[AoD]{angle-of-departure}
\acro{mrt}[MRT]{maximum ratio transmission}
\acro{noma}[NOMA]{non-orthogonal multiple access}
\acro{mse}[MSE]{mean squared error}
\end{acronym}

\bibliographystyle{IEEEtran}
\bibliography{IEEEabrv,references}

\end{document}